\documentclass[preprint,12pt]{article}
\usepackage{bbm}
\usepackage{amsfonts}
\usepackage{mathrsfs}
\usepackage{amssymb}
\usepackage{amsmath}
\usepackage{amsthm}
\usepackage{color}
\usepackage{url}
\usepackage{hyperref}
\usepackage{multirow}
\usepackage{bbding}
\usepackage[ruled,vlined]{algorithm2e}
\usepackage{latexsym}
\usepackage[final]{graphicx}
\usepackage{cases}
\usepackage{amsthm}

\usepackage{lastpage}
\usepackage{epsfig}
\usepackage[misc]{ifsym}
\usepackage{lineno}

\newtheorem{theorem}{Theorem}[section]

\newtheorem{lemma}[theorem]{Lemma}

\theoremstyle{definition}
\newtheorem{definition}[theorem]{Definition}
\newtheorem{remark}[theorem]{Remark}

\begin{document}

\begin{center}

\textbf{\Large{Asymptotically Ideal Hierarchical Secret Sharing Based on CRT for Integer Ring}}

\vspace{0.7cm}

Jian Ding$^{1,2}$, Cheng Wang$^2$, Hongju Li$^{1,3}$, Cheng Shu$^2$, Haifeng Yu$^2$

\vspace{3mm}
\footnotesize{
$^1$School of Mathematics and Big Data, Chaohu University, Hefei 238024, China\\
$^2$School of Mathematics and Statistics, Hefei University, Hefei 230000, China\\
$^3$}School of Computer Science and Technology, University of Science and Technology of China, Hefei 230026, China

\footnotetext{\footnotesize {This research was supported in part by Projects of Chaohu University under Grant
KYQD-202220, Grant 2024cxtd147 and Grant hxkt20250327, in part by University Natural Science Research Project of Anhui Province under Grant 2024AH051324 and Grant 2025AHGXZK30056.}}

\footnotetext{\footnotesize {\noindent Corresponding author: Jian Ding}}

\footnotetext{\footnotesize {E-mail address: dingjian\_happy@163.com (J. Ding)}}
\end{center}

\noindent\textbf{Abstract}: In Shamir's secret sharing scheme, all participants possess equal privileges. However, in many practical scenarios, it is often necessary to assign different levels of authority to different participants. To address this requirement, Hierarchical Secret Sharing (HSS) schemes were developed, which partitioned all participants into multiple subsets and assigned a distinct privilege level to each. Existing Chinese Remainder Theorem (CRT)-based HSS schemes benefit from flexible share sizes, but either exhibit security flaws or have an information rate less than $\frac{1}{2}$. In this work, we propose a disjunctive HSS scheme and a conjunctive HSS scheme by using the CRT for integer ring and one-way functions. Both schemes are asymptotically ideal and are proven to be secure.

\noindent\textbf{Keywords}: disjunctive hierarchical secret sharing, conjunctive hierarchical secret sharing, asymptotically ideal secret sharing, Chinese Remainder Theorem

\section {Introduction}
Secret Sharing (SS) \cite{Shamir1979, Blakley1979} is a method for sharing a secret among a group of participants. It typically consists of two phases: the share generation phase and the secret reconstruction phase. In the share generation phase, a dealer divides the secret into multiple shares and securely distributes each share to its designated participant. In the secret reconstruction phase, any authorized subset can reconstruct the secret by gathering their shares together, while any unauthorized subset cannot. The collection of all authorized subsets is known as the \emph{access structure}. An SS scheme is termed \emph{perfect} if any given unauthorized subset gains no information about the secret. An SS scheme is said to be \emph{ideal} if it is perfect and has an information rate 1. It is referred to as \emph{asymptotically ideal} if it is asymptotically perfect and its information rate approaches 1 as the secret size increases.

The $(t,n)$-threshold scheme is an SS scheme such that any subset of at least $t$ participants is authorized, while any subset with $(t-1)$ or fewer participants is unauthorized. In a $(t,n)$-threshold scheme, all participants possess identical privileges, but in many real-world scenarios, it is often necessary to assign differing privilege levels to different participants. To address this need, Simmons \cite{Simmons1988} introduced the hierarchical threshold access structure in 1988, which partitions the participants into multiple disjoint subsets and assigns a distinct privilege level to each. More precisely, let $\mathcal{P}$ be the set of all participants, which is divided into $u$ disjoint subsets $\mathcal{P}_1, \mathcal{P}_2,\ldots, \mathcal{P}_u$. Denote by $t_1, t_2,\ldots, t_u$ be a sequence of thresholds such that $1\leq t_1<t_2<\cdots<t_u$. By using geometric methods, Simmons \cite{Simmons1988} constructed an SS scheme with access structure
      \[\Gamma_1=\{\mathcal{A} \subseteq \mathcal{P}: \exists \ell \in \{1,2,\ldots,u\}~such~that~|\mathcal{A}\cap(\bigcup_{w=1}^{\ell}\mathcal{P}_{w})|\geq t_{\ell}\},\]
which is called a Disjunctive Hierarchical Secret Sharing (DHSS) scheme. However, this scheme is not ideal and requires the dealer to verify the non-singularity of many matrices. Tassa \cite{Tassa2007} introduced the Conjunctive Hierarchical Secret Sharing (CHSS) scheme, which is an SS scheme with access structure
      \[\Gamma_2=\{\mathcal{A} \subseteq \mathcal{P}: |\mathcal{A}\cap(\bigcup_{w=1}^{\ell}\mathcal{P}_{w})|\geq t_{\ell} ~\mathrm{for~all}~\ell \in \{1,2,\ldots,u\}\}.\]
Moreover, an ideal CHSS scheme and an ideal DHSS scheme were proposed by using Birkhoff interpolation, and the dealers were required to check many matrices. Chen et al. \cite{Chenqi2022} constructed ideal CHSS and DHSS schemes based on polymatroids, requiring the dealer to spend polynomial time to determine specific parameters using specialized algorithms. More recently, Yuan et al. \cite{YuanJiaotong2022} proposed a CHSS scheme based on linear homogeneous recurrence relations and one-way functions. Although their scheme is not ideal due to the use of one-way functions and the publication of many values, it incurs lower computational cost compared to the scheme by Chen et al. \cite{Chenqi2022}.

The Chinese Remainder Theorem (CRT) provides the flexibility to assign shares of varying sizes to different participants. Harn et al. \cite{Harn-Miao2014} constructed a DHSS scheme based on the CRT for integer ring, which was insecure pointed by Ersoy et al. \cite{Oguzhan-Ersoy2016}. Furthermore, Ersoy et al. \cite{Oguzhan-Ersoy2016} proposed a new DHSS scheme and a new CHSS scheme by using the CRT for integer ring and hash functions. Their information rates were smaller than $\frac{1}{2}$. Tiplea et al. \cite{Tiplea2021} proposed an asymptotically ideal DHSS scheme and an asymptotically ideal CHSS scheme by using the CRT for integer ring, achieving information rates approaching 1. Recently, Yang et al. \cite{Yangjing2024} constructed an ideal DHSS scheme by using the CRT for polynomial rings. However, all these CRT-based ideal or asymptotically ideal DHSS schemes and CHSS schemes \cite{Tiplea2021,Yangjing2024} have been shown to be insecure in our other work \cite{Hongjuarx}.

\emph{Our contributions}. By using the CRT for integer ring and one-way functions, we construct two asymptotically ideal hierarchical secret sharing schemes with flexible share sizes. More precisely, we make the following contributions:
\begin{itemize}
  \item[(1)] We construct an asymptotically ideal DHSS scheme, which has flexible share sizes, but the ideal DHSS schemes of Tassa et al. \cite{Tassa2007} and Cheng et al.\cite{Chenqi2022} cannot. Compared with the CRT-based DHSS schemes of Harn et al. \cite{Harn-Miao2014}, Ersoy et al. \cite{Oguzhan-Ersoy2016}, Tiplea et al. \cite{Tiplea2021} and Yang et al. \cite{Yangjing2024}, our scheme is a secure and asymptotically ideal DHSS scheme (See Table 1).
\begin{table}[!htb]
  \centering
   {\bf Table 1.}  Disjunctive hierarchical secret sharing schemes. \\
  \begin{tabular}{cccccc}
  \hline
    \multirow{2}*{Schemes}  & \multirow{2}*{Secure}    & Accommodate             &\multirow{2}*{Perfectness}  &Information\\
    ~                        &~                        &flexible share sizes      &~                          &rate $\rho$\\
  \hline
  \hline
  \cite{Tassa2007}              &Yes                      &No                      &Yes                       &$\rho=1$\\
  \hline
  \cite{Chenqi2022}             &Yes                      &No                      &Yes                       &$\rho=1$\\
  \hline
  \cite{Harn-Miao2014}          &No                       &Yes                     &No                        &$\rho<1$\\
  \hline
   \cite{Tiplea2021}            &No                       &Yes                     &No                        &$\rho<1$\\
  \hline
  \cite{Yangjing2024}           &No                       &Yes                     &No                        &$\rho=1$\\
  \hline
  \cite{Oguzhan-Ersoy2016}      &Yes                      &Yes                   &Asymptotic                   &$\rho<\frac{1}{2}$\\
  \hline
  Our scheme                    &Yes                      &Yes                   &Asymptotic                  &Approaches~1\\
    \hline
  \end{tabular}
 \end{table}

 \item[(2)] We construct an asymptotically ideal CHSS scheme, which has flexible share sizes, but the ideal CHSS schemes of Tassa et al. \cite{Tassa2007}, Cheng et al.\cite{Chenqi2022} and the asymptotically ideal CHSS scheme of Yuan et al. \cite{YuanJiaotong2022} cannot. Compared with the CRT-based CHSS schemes of Tiplea et al. \cite{Tiplea2021} and Ersoy et al. \cite{Oguzhan-Ersoy2016}, our scheme is a secure and asymptotically ideal CHSS scheme (See Table 2).
     \begin{table}[!htb]
  \centering
   {\bf Table 2.}  Conjunctive hierarchical secret sharing schemes. \\
  \begin{tabular}{cccccc}
  \hline
    \multirow{2}*{Schemes}  & \multirow{2}*{Secure}    & Accommodate             &\multirow{2}*{Perfectness}  &Information\\
    ~                        &~                        &flexible share sizes      &~                          &rate $\rho$\\
  \hline
  \hline
  \cite{Tassa2007}              &Yes                      &No                      &Yes                       &$\rho=1$\\
  \hline
  \cite{Chenqi2022}             &Yes                      &No                      &Yes                       &$\rho=1$\\
  \hline
  \cite{YuanJiaotong2022}       &Yes                       &No                     &Asymptotic                &$\rho=1$\\
  \hline
   \cite{Tiplea2021}            &No                       &Yes                     &No                        &Approaches 1\\
  \hline
  \cite{Oguzhan-Ersoy2016}      &Yes                      &Yes                  &Asymptotic                   &$\rho<\frac{1}{2}$\\
  \hline
  Our scheme                    &Yes                      &Yes                   &Asymptotic                  &Approaches~1\\
  \hline
  \end{tabular}
 \end{table}

\end{itemize}

\emph{Paper organization}. The remainder of this paper is structured as follows. Section 2 introduces the necessary preliminaries. Section 3 presents the construction of our DHSS scheme, along with its security analysis. In section 4, we construct a CHSS scheme and provide its security analysis. Section 5 concludes the paper.

\section{Preliminaries}\label{Sec: Prelim}
This section will introduce the relevant mathematical tools required in this paper, the related concepts of secret sharing, and the Asmuth-Bloom secret sharing scheme \cite{Asmuth-Bloom1983}. Denote by $[n]=\{1,2,\ldots,n\}$ and $[N_1, N_2]=\{N_1,N_1+1,\ldots,N_2\}$, where $n, N_1,N_2$ be positive integers such that $N_1<N_2$. Let $\mathbb{Z}$ denote the set of all integers, and $\mathbb{Z}_n$ the ring that integers modulo $n$.

\subsection{Relevant mathematical tools}
\begin{lemma}[CRT for integer ring, \cite{CRTdefinition}]\label{le: CRT}
Let $m_1,m_2,\ldots,m_n$ be positive integers that are pairwise coprime. Given any integers $x_{1},x_{2},\ldots,x_{n}$ and a system of congruences
 \begin{displaymath}
         \left\{\begin{aligned}
            x &\equiv x_{1} \pmod{m_{1}},\\
            x &\equiv x_{2} \pmod{m_{2}},\\
              &~\vdots\\
            x &\equiv  x_{n} \pmod{m_{n}},\\
           \end{aligned} \right.
        \end{displaymath}
it holds that
                  \[x\equiv\sum\limits_{i=1}^{n}\lambda_iM_i x_i \pmod {M},\]
where $M=\prod\limits_{i=1}^{n}m_i$, $M_i=\frac{M}{m_i}$, and $\lambda_i\equiv M_i^{-1}\pmod {m_i}$.
If the integer $x$ satisfies $0 \leq x < M$, the solution is unique, denoted by
                 \[x=\sum\limits_{i=1}^{n}\lambda_iM_i x_i \pmod {M}.\]
\end{lemma}

\begin{definition}[$k$-compact sequence of co-primes \cite{Tiplea2018}]\label{def:k-compact sequences}
Let $L=\{m_0,m_1,\ldots,m_n\}$ be a sequence of co-primes.
\begin{itemize}
  \item[(1)] The sequence $L$ is called $(k,\theta)$-compact, where $k\geq 1$ and $\theta \in (0,1)$ are real numbers, if $m_{0}<m_{1}<\cdots<m_{n}$, and $km_{0}<m_{i}<km_{0}+m_{0}^{\theta}$ for all $i \in [n]$.

  \item[(2)] The sequence $L$ is called $k$-compact if it is $(k,\theta)$-compact for some $\theta \in (0,1)$.
\end{itemize}
\end{definition}

\begin{remark}
A sequence of co-primes $L=\{m_0,m_{1},\ldots,m_{n}\}$ is $k$-compact if and only if the sub-sequences $L_{j}=\{m_0,m_{1},\ldots,m_{j}\}, j \in [n]$ are $k$-compact.
\end{remark}
\subsection{Secret sharing}
For discrete random variables $\mathbf{X}$ and $\mathbf{Y}$, let $\mathsf{H}(\mathbf{X})$ denote the Shannon entropy of $\mathbf{X}$, and $\mathsf{H}(\mathbf{X} \mid \mathbf{Y})$ the conditional entropy of $\mathbf{X}$ given $\mathbf{Y}$.

\begin{definition}[Secret sharing scheme]
 Let $\mathcal{P}=\{P_1,P_2,\ldots,P_n\}$ be a set of $n$ participants. A secret sharing scheme consists of a share generation phase and a secret reconstruction phase, as follows.
\begin{itemize}
  \item[(1)] Share Generation Phase: Let $\mathcal{S}$ be the secret space. For any secret $s \in \mathcal{S}$, the dealer distributes shares to participants using a share distribution algorithm
                      \[\mathsf{SHARE}\colon\mathcal{S}\times\mathcal{R}\mapsto\mathcal{S}_1\times\mathcal{S}_2\times\dotsb\times\mathcal{S}_n,\]
  where $\mathcal{R}$ is a set of random symbol set, and $\mathcal{S}_i$ is the share space of the participant $P_i$.
  \item[(2)] Secret Reconstruction Phase. Participants in any authorized set  $\mathcal{A}\subseteq \mathcal{P}$ can use their shares to reconstruct the secret through the secret reconstruction algorithm
                        \[\mathsf{RECON}\colon\prod_{P_i\in\mathcal{A}}\mathcal{S}_i\mapsto\mathcal{S}.\]
  Conversely, any unauthorized set cannot reconstruct the secret.
\end{itemize}
\end{definition}

In this work, the index $i$ denotes the $i$-th participant $P_i$, and the set $[n]$ corresponds to the participant set $\mathcal{P}$. A secret sharing scheme is said to be perfect if no unauthorized subset can obtain any information about the secret.

\begin{definition}[Information rate, \cite{Ning2018}]\label{def: Information rate}
The information rate of a secret sharing scheme is defined as
                               \[\rho=\frac{\log_2|\mathcal{S}|}{\max_{i\in [n]}^{}{\log_2|\mathcal{S}_i|}},\]
where $|\mathcal{S}|$ is the size of the secret space, and $|\mathcal{S}_i|$ is the size of the share space of participant $P_i$.
\end{definition}

For any perfect secret sharing scheme, the information rate satisfies $\rho \leq 1$, and a perfect scheme with $\rho = 1$ is called ideal.

\begin{definition}[Disjunctive hierarchical secret sharing scheme, \cite{Simmons1988}]\label{def: Multilevel secret sharing of disjunctive}
Let $\mathcal{P}$ be the set of all participants, partitioned into $u$ disjoint subsets $\mathcal{P}_1, \mathcal{P}_2,\ldots, \mathcal{P}_u$, i.e.,
 \[\mathcal{P}=\cup_{\ell=1}^{u} \mathcal{P}_{\ell},~and~\mathcal{P}_{\ell_1}\cap\mathcal{P}_{\ell_2}=\varnothing~ for~any~ 1\leq \ell_1<\ell_2\leq u.\]
Denote by $t_1, t_2,\ldots, t_u$ be a sequence of thresholds such that $1\leq t_1<t_2<\cdots<t_u$. A Disjunctive Hierarchical Secret Sharing (DHSS) scheme is a secret sharing scheme with the following properties:
\begin{itemize}
  \item[(1)] Correctness. Any subset
       \[\mathcal{A} \in\Gamma_1=\{\mathcal{A} \subseteq \mathcal{P}: \exists \ell \in [u]~such~that~|\mathcal{A}\cap(\bigcup_{w=1}^{\ell}\mathcal{P}_{w})|\geq t_{\ell}\}.\]
       can reconstruct the secret.
  \item[(2)] Privacy. No subset $\mathcal{B}\notin \Gamma_1$ can reconstruct the secret.
\end{itemize}
A DHSS scheme is ideal if it is perfect and achieves information rate one.
\end{definition}

\begin{definition}[Conjunctive hierarchical secret sharing scheme, \cite{Tassa2007}]\label{def: Multilevel secret sharing of conjunctive}
The setting is the same as in Definition \ref{def: Multilevel secret sharing of disjunctive}, but the access structure $\Gamma_1$ is replaced by
\begin{displaymath}
       \Gamma_2=\{\mathcal{A} \subseteq \mathcal{P}: |\mathcal{A}\cap(\bigcup_{w=1}^{\ell}\mathcal{P}_{w})|\geq t_{\ell} ~for~ \forall \ell \in [u]\}.
\end{displaymath}
\end{definition}

\begin{definition}[Asymptotically ideal hierarchical secret sharing scheme, \cite{Quisquater2002}]\label{def: Asymptotically Ideal HSS}
Considetr a hierarchical secret sharing (either DHSS or CHSS) scheme with the secret space $\mathcal{S}$ and share spaces $\mathcal{S}_i, i\in [n]$. The scheme is said to be asymptotically ideal if it satisfies the following two asymptotic conditions.
\begin{itemize}
  \item[(1)] Asymptotic perfectness. For all $\epsilon_1>0$, there is a positive integer $\sigma_1$ such that for all $\mathcal{B}\notin\Gamma$ and $|\mathcal{S}|>\sigma_1$, the loss entropy
                       \[\Delta(|\mathcal{S}|)=\mathsf{H}(\mathbf{S})-\mathsf{H}(\mathbf{S}|\mathbf{V}_{\mathcal{B}})\leq \epsilon_1,\]
   where $\mathsf{H}(\mathbf{S})\neq 0$, and $\mathbf{S}, \mathbf{V}_{\mathcal{B}}$ are random variables corresponding to the secret and the knowledge of $\mathcal{B}$, respectively.
  \item[(2)] Asymptotic maximum information rate. For all $\epsilon_2>0$, there is a positive integer $\sigma_2$ such that for all $\mathcal{B}\notin\Gamma$ and $|\mathcal{S}|>\sigma_2$, it holds that
      \[\frac{\max_{i\in [n]}^{}{\mathsf{H}(\mathbf{S}_i)}}{\mathsf{H}(\mathbf{S})}\leq 1+\epsilon_2.\]
\end{itemize}
\end{definition}

\subsection{Dr$\breve{\mathrm{a}}$gan-Tiplea secret sharing scheme}
Based on the Chinese Remainder Theorem (CRT) for integer ring, the Asmuth-Bloom secret sharing scheme \cite{Asmuth-Bloom1983} enables any $t$ participants to reconstruct the secret, while any $t-1$ participants cannot. However, this scheme is not perfect and achieves an information rate less than $1$. To address these limitations, Dr$\breve{\mathrm{a}}$gan and Tiplea \cite{Tiplea2018} proposed a new secret sharing scheme using the CRT for integer ring along with a $k$-compact sequence of pairwise co-prime integers. Their scheme, referred to as the Dr$\breve{\mathrm{a}}$gan-Tiplea secret sharing scheme, can be viewed as a generalization of the Asmuth-Bloom approach. More details are given below.

Let $m_0$ be a positive integer and denote by $\mathbb{Z}_{m_0}$ the ring of integers modulo $m_0$. The Dr$\breve{\mathrm{a}}$gan-Tiplea secret sharing scheme is a secret sharing scheme consisting of two phases.
\begin{itemize}
  \item[(1)] Share Generation Phase: For any secret $s\in \mathcal{S}=\mathbb{Z}_{m_{0}}$, the dealer generates and publishes a $k$-compact sequence of co-primes $L=\{m_{0},m_{1},
       \ldots,m_{n}\}$, namely, the following three conditions are satisfied.
          \begin{itemize}
            \item[(i)] $m_0, m_1,\ldots, m_n$ are pairwise co-primes.
             \item[(ii)] $m_{0}<m_{1}<\cdots<m_{n}$.
             \item[(iii)]$km_{0}<m_{i}<km_{0}+m_{0}^{\theta}$ for all $i \in [n]$, where $k \geq 1$ and $\theta \in (0,1)$ are real numbers.
           \end{itemize}
    The dealer then chooses an integer $\alpha$ such that $0\leq y=s+\alpha m_{0}<\prod_{i=1}^{t}m_{i}$. For each $i\in [n]$, the share $s_i=y \pmod {m_{i}}$ is computed and securely distributed to the participant $P_i$.
  \item[(2)] Secret Reconstruction Phase: Given any $\mathcal{A}\subseteq [n]$ of size $|\mathcal{A}|\geq t$, participants from $\mathcal{A}$ establish a system of congruences
                \[y\equiv s_{i} \pmod{m_{i}},i\in \mathcal{A}.\]
    Since $y<\prod_{i=1}^{t}m_{i}<\prod_{i\in \mathcal{A}}m_{i}$, then $y=s+\alpha m_{0}$ is determined by using Lemma \ref{le: CRT}, and $s=y \pmod {m_0}$.
\end{itemize}

\begin{remark}
The Dr$\breve{\mathrm{a}}$gan-Tiplea secret sharing scheme is asymptotically ideal if the $k$-compact sequence of pairwise co-prime integers is a $1$-compact sequence of co-primes (see Corollary 1 of \cite{Tiplea2018}). Moreover, such $1$-compact sequences can be generated efficiently (see Section $5$ of \cite{Tiplea2021}).
\end{remark}

\begin{lemma}[\cite{Tiplea2018}]\label{le: resutls-Tiplea2018}
In the Dr$\breve{\mathrm{a}}$gan-Tiplea secret sharing scheme, let $\mathcal{B}\subset \mathcal{P}, |\mathcal{B}|=t-1$ and $\eta_\mathcal{B}=\left\lfloor\left.{\prod\limits_{i=1}^{t}m_i}\right/{\prod\limits_{i\in\mathcal{B}}m_i}\right\rfloor$, where $\lfloor \cdot\rfloor$ is the floor function. Denote by
      \[\mathcal{G}_\mathcal{B}=\{g\in\mathbb{Z}:~L~\mathrm{is}~k-\mathrm{compact}, 0\leq g <\prod_{i=1}^{t}m_{i},~\mathrm{and}~g\equiv s_i \pmod{m_i}, i\in \mathcal{B}\}.\]
The following results hold.
\begin{itemize}
  \item[(1)] $|\mathcal{G}_\mathcal{B}|\in \{\eta_\mathcal{B},\eta_\mathcal{B}+1\}$. (see Lemma 2 in \cite{Tiplea2018})

  \item[(2)] Define the mapping $\Psi_{\mathcal{B}}$ by
 \[\Psi_{\mathcal{B}}: \mathcal{G}_{\mathcal{B}}\mapsto \mathcal{S}, g\mapsto g\pmod{m_0}.\]\
  For any $s\in \mathcal{S}$, let
           \[\Psi_{\mathcal{B}}^{-1}(s)=\{g\in \mathcal{G}_{\mathcal{B}}: g\equiv s\pmod{m_0}\},\]
  then $|\Psi_{\mathcal{B}}^{-1}(s)|=\left\lfloor \frac{|\mathcal{G}_\mathcal{B}|}{m_0}\right\rfloor+a_s$ for some $a_s\in \mathbb{Z}_2$. Let $\delta_0$ and $\delta_1$ be the sizes of the sets $\mathcal{G}_{\mathcal{B},0}$ and $\mathcal{G}_{\mathcal{B},1}$, respectively, where
         \begin{displaymath}
             \begin{aligned}
             &\mathcal{G}_{\mathcal{B},0}=\left\{s\in \mathcal{S}:|\Psi_{\mathcal{B}}^{-1}(s)|=\left\lfloor \frac{|\mathcal{G}_\mathcal{B}|}{m_0}\right\rfloor\right\},\\
             &\mathcal{G}_{\mathcal{B},1}=\left\{s\in \mathcal{S}:|\Psi_{\mathcal{B}}^{-1}(s)|=\left\lfloor \frac{|\mathcal{G}_\mathcal{B}|}{m_0}\right\rfloor+1\right\},
             \end{aligned}
        \end{displaymath}
  then $\delta_1=|\mathcal{G}_{\mathcal{B}}|\pmod{m_0}, \delta_0=m_0-\delta_1$ and
            \[|\mathcal{G}_\mathcal{B}|=\delta_0\left\lfloor \frac{|\mathcal{G}_\mathcal{B}|}{m_0}\right\rfloor+\delta_1\left(\left\lfloor \frac{|\mathcal{G}_\mathcal{B}|}{m_0}\right\rfloor+1\right).~(\mathrm{see~the~proof~of~Lemma~2~in}~\cite{Tiplea2018})\]

  \item[(3)] $\lim\limits_{m_0\to\infty}\frac{\prod\limits_{i=1}^{t}m_i}{m_0\prod\limits_{i\in\mathcal{B}}m_i}=k$ (see the proof of Theorem 1 in \cite{Tiplea2018}). Based on this, it has that for sufficiently large $m_0$, there are only two cases need to be considered. (see the proof of Theorem 1 in \cite{Tiplea2018})
     \begin{itemize}
       \item[(i)] When $km_0-(m^{\theta}_0+1)<|\mathcal{G}_\mathcal{B}|<km_0$, it has that $\left\lfloor \frac{|\mathcal{G}_\mathcal{B}|}{m_0}\right\rfloor=k-1$, $\lim\limits_{m_0\to\infty}\frac{\delta_0}{m_0}=0$, and $\lim\limits_{m_0\to\infty}\frac{\delta_1}{m_0}=1$.
       \item[(ii)] When $km_0\leq|\mathcal{G}_\mathcal{B}|<km_0+(m^{\theta}_0+1)$, it has that $\left\lfloor \frac{|\mathcal{G}_\mathcal{B}|}{m_0}\right\rfloor=k$, and
       $0\leq\delta_1<m^{\theta}_0+1$, which means that $\lim\limits_{m_0\to\infty}\frac{\delta_0}{m_0}=1, \lim\limits_{m_0\to\infty}\frac{\delta_1}{m_0}=0$.
     \end{itemize}
 In other words, there is an integer $b\in \mathbb{Z}_{2}$ such that
     \[\lim\limits_{m_0\to\infty}\frac{\delta_b}{m_0}=1, \lim\limits_{m_0\to\infty}\frac{{\sum\limits_{q\in \mathbb{Z}_{2},q\neq b}}\delta_q}{m_0}=0.\]
  \item[(4)] When $m_0$ goes to infinity, the loss entropy
           \begin{displaymath}
             \begin{aligned}
             \Delta(|\mathcal{S}|)
             &=\mathsf{H}(\mathbf{S})-\mathsf{H}(\mathbf{S}|\mathbf{V}_{\mathcal{B}})\\
             &=\log_2{m_0}+\delta_0\frac{\left\lfloor \frac{|\mathcal{G}_\mathcal{B}|}{m_0}\right\rfloor}{|\mathcal{G}_\mathcal{B}|}\log_2\frac{\left\lfloor \frac{|\mathcal{G}_\mathcal{B}|}{m_0}\right\rfloor}{|\mathcal{G}_\mathcal{B}|}+\delta_1\frac{\left\lfloor \frac{|\mathcal{G}_\mathcal{B}|}{m_0}\right\rfloor+1}{|\mathcal{G}_\mathcal{B}|}\log_2\frac{\left\lfloor \frac{|\mathcal{G}_\mathcal{B}|}{m_0}\right\rfloor+1}{|\mathcal{G}_\mathcal{B}|}\\
              \end{aligned}
          \end{displaymath}
  goes to $0$, where $\mathsf{H}(\mathbf{S})\neq 0$, and $\mathbf{S}, \mathbf{V}_{\mathcal{B}}$ are random variables corresponding to the secret and the shares of $\mathcal{B}$, respectively. (see the proof of Lemma 2 in \cite{Tiplea2018})

\end{itemize}
\end{lemma}


\section{A novel asymptotically ideal DHSS scheme}\label{Sec: our scheme}
In this section, we construct a hierarchical secret sharing scheme for the disjunctive access structure. The construction is based on the CRT for integer ring and one-way functions. We present our scheme in Subsection \ref{Subsec:our DHSS} and analyze its security in Subsection \ref{Subsec: security of our DHSS}.

\subsection{Our scheme 1}\label{Subsec:our DHSS}
Let $\mathcal{P}$ be a set of $n$ participants, divided into $u$ pairwise disjoint subsets $\mathcal{P}_1, \mathcal{P}_2, \dots, \mathcal{P}_u$. For each $\ell\in[u]$, define $n_{\ell}=|\mathcal{P}_{\ell}|$ and let $N_{\ell}=\sum_{w=1}^{\ell} n_{w}$. Additionally, for every $\ell\in[u]$, denote by $t_{\ell}$ the threshold assigned to the union $\bigcup_{w=1}^{\ell}\mathcal{P}_w$. The thresholds satisfy $1\leq t_1<t_2<\cdots<t_u$ and $t_{\ell}\leq N_{\ell}$ for all $\ell\in [u]$. Our scheme 1 consists of two phases: the share generation phase and the secret reconstruction phase.

\textbf{(1) Share Generation Phase.} Let $m_{0}$ be a big positive integer, and let $\mathcal{S}=\mathbb{Z}_{m_{0}}$ be the secret space.
 \begin{itemize}
       \item[Step 1.] The dealer generates and publishes a $k$-compact sequence of co-primes $L=\{m_{0},m_{1},\\
       \ldots,m_{n}\}$, namely, the following three conditions are satisfied.
          \begin{itemize}
            \item[(i)] $m_0, m_1,\ldots, m_n$ are pairwise co-primes.
             \item[(ii)] $m_{0}<m_{1}<\cdots<m_{n}$.
             \item[(iii)]$km_{0}<m_{i}<km_{0}+m_{0}^{\theta}$ for all $i \in [n]$, where $k \geq 1$ and $\theta \in (0,1)$ are real numbers.
           \end{itemize}

       \item[Step 2.] For any given secret $s \in \mathcal{S}$, the dealer selects random integers $c_{i} \in \mathbb{Z}_{m_{i}}$, $i \in [N_{u-1}]$, and
           selects integers $\alpha_{\ell}$, $\ell\in [u]$ such that
           \[0 \leq y_{\ell}=s+\alpha_{\ell}m_{0} <\prod\limits_{i=1}^{t_{\ell}}m_i, \ell\in [u].\]

           The dealer sends the share $s_i$ to the $i$-th participant, where
                \[s_i=\begin{cases}
                   c_i,\mathrm{if}~i\in [N_{u-1}],\\
                 y_u \pmod {m_i},\mathrm{if}~i\in [N_{u-1}+1,N_u].\\
                \end{cases}\]

       \item[Step 3.] The dealer selects $u$ publicly known distinct one-way functions $h_{\ell}, \ell\in [u]$. Each function accepts an input of arbitrary length and produces an output of fixed length $\lfloor \log_{2} m_n \rfloor$. Then the dealer publishes the values
                    \[w_i^{(\ell)}=(y_{\ell}-h_{\ell}(s_i))\pmod{m_i}, \ell\in [u-1], i\in [N_{\ell}],\]
           and $w_i^{(u)}=(y_{u}-h_{u}(s_i))\pmod{m_i}$, $i\in [N_{u-1}]$.
   \end{itemize}

\textbf{(2) Secret Reconstruction Phase}. For any $\mathcal{A} \subseteq \mathcal{P}$ such that $\mathcal{A}^{(\ell)}=\mathcal{A}\cap (\cup_{w=1}^{\ell}\mathcal{P}_{w}),\\
 |\mathcal{A}^{(\ell)}|\geq t_{\ell}$ for some $\ell\in [u]$, participants of $\mathcal{A}^{(\ell)}$ compute
         \[s_i^{(\ell)}=\begin{cases}
           (h_{\ell}(s_i)+w_i^{(\ell)})\pmod{m_i},~\mathrm{if}~\ell\in [u-1], i\in \mathcal{A}^{(\ell)}\subseteq [N_{\ell}],\\
            (h_{u}(s_i)+w_i^{(u)}) \pmod{m_i},~\mathrm{if}~\ell=u, i\in \mathcal{A}^{(\ell)}\cap [N_{u-1}],\\
            s_i,\mathrm{if}~\ell=u, i\in \mathcal{A}^{(\ell)}\cap [N_{u-1}+1,N_u],\\
         \end{cases}\]
and determine
                          \[y_{\ell}=\sum\limits_{i\in \mathcal{A}^{(\ell)}}\lambda_{i,\mathcal{A}^{(\ell)}}M_{i,\mathcal{A}^{(\ell)}}s_i^{(\ell)} \pmod {M_{\mathcal{A}^{(\ell)}}},\]
where $M_{\mathcal{A}^{(\ell)}}=\prod\limits_{i\in\mathcal{A}^{(\ell)}}m_i$, $M_{i,\mathcal{A}^{(\ell)}}=M_{\mathcal{A}^{(\ell)}}/m_i$, $\lambda_{i,\mathcal{A}^{(\ell)}}\equiv M_{i,\mathcal{A}^{(\ell)}}^{-1}\pmod {m_i}$. Consequently, the secret is reconstructed as $s=y_{\ell}\pmod{m_0}$.

\subsection{Security analysis of our scheme 1}\label{Subsec: security of our DHSS}
We will prove the correctness, asymptotic perfectness and asymptotic maximum information rate of our scheme 1 in this subsection.
\begin{theorem}[Correctness]\label{Theorem:Correctness of our scheme1} Any subset
       \[\mathcal{A} \in\Gamma_1=\{\mathcal{A} \subseteq \mathcal{P}: \exists \ell \in [u]~such~that~|\mathcal{A}\cap(\bigcup_{w=1}^{\ell}\mathcal{P}_{w})|\geq t_{\ell}\}.\]
       can reconstruct the secret.
\end{theorem}
\begin{proof}
Suppose that there is an integer $\ell\in [u]$ such that $\mathcal{A}^{(\ell)}=\mathcal{A}\cap (\cup_{w=1}^{\ell}\mathcal{P}_{w})$ and $|\mathcal{A}^{(\ell)}|\geq t_{\ell}$. Each participant $P_i$ in $\mathcal{A}^{(\ell)}$ computes $s_i^{(\ell)}$ from its own share $s_i$ as follows:
         \[s_i^{(\ell)}=\begin{cases}
           (h_{\ell}(s_i)+w_i^{(\ell)})\pmod{m_i},~\mathrm{if}~\ell\in [u-1], i\in \mathcal{A}^{(\ell)}\subseteq [N_{\ell}],\\
            (h_{u}(s_i)+w_i^{(u)}) \pmod{m_i},~\mathrm{if}~\ell=u, i\in \mathcal{A}^{(\ell)}\cap [N_{u-1}],\\
            s_i,\mathrm{if}~\ell=u, i\in \mathcal{A}^{(\ell)}\cap [N_{u-1}+1,N_u].\\
         \end{cases}\]
Here $h_{\ell}, \ell\in [u]$ are public one-way functions. Moreover, the values
\[w_i^{(\ell)}=(y_{\ell}-h_{\ell}(s_i))\pmod{m_i}, \ell\in [u-1], i\in [N_{\ell}],\]
           and $w_i^{(u)}=(y_{u}-h_{u}(s_i))\pmod{m_i}$, $i\in [N_{u-1}]$ are publicly known. These values give rise to a system of congruences
      \[y_{\ell} \equiv s_{i}^{(\ell)} \pmod{m_{i}}, i\in \mathcal{A}^{(\ell)}.\]

Since $m_{0}<m_{1}<\cdots<m_{n}$ and $|\mathcal{A}^{(\ell)}|\geq t_{\ell}$, then
                \[y_{\ell}<\prod\limits_{i=1}^{t_{\ell}}m_i\leq\prod\limits_{i\in \mathcal{A}^{(\ell)}}m_i.\]
Therefore, by the Chinese Remainder Theorem for integer ring (Lemma~\ref{le: CRT}), the above system has a unique solution
   \[y_{\ell}=\sum\limits_{i\in \mathcal{A}^{(\ell)}}\lambda_{i,\mathcal{A}^{(\ell)}}M_{i,\mathcal{A}^{(\ell)}}s_i^{(\ell)} \pmod {M_{\mathcal{A}^{(\ell)}}},\]
where $M_{\mathcal{A}^{(\ell)}}=\prod\limits_{i\in\mathcal{A}^{(\ell)}}m_i$, $M_{i,\mathcal{A}^{(\ell)}}=M_{\mathcal{A}^{(\ell)}}/m_i$ and $\lambda_{i,\mathcal{A}^{(\ell)}} \equiv M_{i,\mathcal{A}^{(\ell)}}^{-1}\pmod {m_i}$. Consequently, the secret is reconstructed as $s=y_{\ell}\pmod{m_0}$.
\end{proof}

Now we prove the \emph{asymptotic perfectness} of our scheme 1. Let $\mathcal{B}\subset \mathcal{P}$, and
            \[\mathcal{B}\notin \Gamma_1=\{\mathcal{A} \subseteq \mathcal{P}: \exists \ell \in [u]~such~that~|\mathcal{A}\cap(\bigcup_{w=1}^{\ell}\mathcal{P}_{w})|\geq t_{\ell}\}.\]
Consider the worst case, i.e.,
    \[|\mathcal{B}\cap(\cup_{w=1}^{\ell}\mathcal{P}_{w})|=t_{\ell}-1~\mathrm{for~all}~\ell\in [u].\]
The participants in $\mathcal{B}$ have access to their own shares, the upper and lower bounds of $y_{\ell}$ for all $\ell \in [u]$, as well as all publicly available information. Consequently, they attempt to reconstruct the secret by first choosing a tuple $(g_1, g_2, \ldots, g_u)\in \mathbb{Z}^{u}$ that meets the following five conditions, and then calculating $g_u \pmod{m_0}$.
\begin{itemize}
  \item[(i)] For each $\ell\in [u]$, the value $g_{\ell}$ satisfies $0\leq g_{\ell}<\prod\limits_{i=1}^{t_{\ell}}m_i$.
  \item[(ii)] The congruences $g_{1}\equiv g_{2}\equiv\cdots\equiv g_{u}\pmod{m_0}$ hold.
  \item[(iii)] The following congruences are satisfied:
      \[g_{\ell}\equiv(w_i^{(\ell)}+h_{\ell}(s_i))\pmod{m_i}, \ell\in [u-1], i\in \mathcal{B}\cap[N_{\ell}],\]
      and $g_{u}\equiv(w_i^{(u)}+h_{u}(s_i))\pmod{m_i}$, $i\in \mathcal{B}\cap[N_{u-1}]$, that is
             \begin{displaymath}
         \left\{\begin{aligned}
            g_{1}\equiv &(w_i^{(1)}+h_{1}(s_i))\pmod{m_i}~\mathrm{for~all}~i\in \mathcal{B}\cap [N_1],\\
            g_{2}\equiv &(w_i^{(2)}+h_{2}(s_i))\pmod{m_i}~\mathrm{for~all}~i\in \mathcal{B}\cap [N_2],\\
                     &\vdots\\
            g_{u-1}\equiv &(w_i^{(u-1)}+h_{u-1}(s_i))\pmod{m_i}~\mathrm{for~all}~i\in \mathcal{B}\cap [N_{u-1}],\\
            g_{u}\equiv &(w_i^{(u)}+h_{u}(s_i))\pmod{m_i}~\mathrm{for~all}~i\in \mathcal{B}\cap [N_{u-1}].\\
           \end{aligned} \right.
        \end{displaymath}
  \item[(iv)] For $i\in \mathcal{B}\cap [N_{u-1}+1, N_u]$, it holds that $g_{u}\equiv s_i\pmod{m_i}$.
  \item[(v)] For any $i\in [N_{u-1}]$ with $i\notin \mathcal{B}$, suppose $i$ belongs to level $\mathcal{P}_{\ell_1}$. Then there exists an integer $\widetilde{s}_i\in \mathbb{Z}_{m_i}$ such that for all $\ell\in [\ell_1,u]$, the relation $h_{\ell}(\widetilde{s}_i)\equiv(g_{\ell}-w_i^{(\ell)})\pmod{m_i}$ holds.
  \end{itemize}

\begin{lemma}\label{le: loss entropy of DHSS}
Let $\mathcal{V}_{\mathcal{B}}$ denote the set of conditions (i) through (v) that are the knowledge of $\mathcal{B}$. Denote by $\mathcal{V}^\prime_{\mathcal{B}}$ the set consisting of conditions (i) through (iv). For any $\epsilon_1>0$, there exists a positive integer $\sigma_1$ such that whenever $|\mathcal{S}|=m_0>\sigma_1$, the following holds:
                               \[0<\mathsf{H}(\mathbf{S}|\mathbf{V}^\prime_{\mathcal{B}})-\mathsf{H}(\mathbf{S}|\mathbf{V}_{\mathcal{B}})<\epsilon_1,\]
where $\mathbf{V}^\prime_{\mathcal{B}}$ and $\mathbf{V}_{\mathcal{B}}$ are random variables corresponding to $\mathcal{V}^\prime_{\mathcal{B}}$ and $\mathcal{V}_{\mathcal{B}}$, respectively.
\end{lemma}
\begin{proof}
Since the dealer selects the shares $s_i$(for $i\in [N_{u-1}]$) uniformly at random and the one-way functions $h_{\ell}(\cdot), \ell\in [u]$ are distinct, then the condition (v) eliminates a candidate tuple $(g_1,g_2,\ldots,g_u)\in \mathbb{Z}^u$ with negligible probability, as long as the secret space $|\mathcal{S}|$ is sufficiently large. Hence, for sufficiently large $|\mathcal{S}|$, the difference in entropy
$\mathsf{H}(\mathbf{S}|\mathbf{V}^\prime_{\mathcal{B}})-\mathsf{H}(\mathbf{S}|\mathbf{V}_{\mathcal{B}})$ becomes negligible. This completes the proof.
\end{proof}

Let us define
\begin{equation}\label{Eq: definiton F}
\mathcal{G}_1=\{(g_1,g_2,\ldots,g_u)\in \mathbb{Z}^u: ~\text{conditions (i)~through~(iv) hold}\}.
\end{equation}
Based on Lemma \ref{le: preimage}, we give the size of $\mathcal{G}_1$ in Theorem \ref{th: cardinality of G}. By combining Lemmas \ref{le: The limits} and \ref{le:bound}, we evaluate the conditional entropy $\mathsf{H}(\mathbf{S}|\mathbf{V}^\prime_{\mathcal{B}})$ and consequently prove that our scheme 1 achieves asymptotic perfectness in Theorem \ref{Th: our perfectness1}.

\begin{lemma}\label{le: preimage}
Define the mapping $\Psi_1$ by
 \[\Psi_1: \mathcal{G}_1\mapsto \mathcal{S}, (g_1,g_2,\ldots,g_u)\mapsto g_u\pmod{m_0}.\]
 For any $s\in \mathcal{S}$, let
           \[\Psi_1^{-1}(s)=\{(g_1,g_2,\ldots,g_u)\in \mathcal{G}_1: g_{u}\equiv s\pmod{m_0}\},\]
then the size of the set $\Psi_1^{-1}(s)$ can be written as
            \[|\Psi_1^{-1}(s)|=\prod_{\ell=1}^u \left(\left\lfloor{\prod\limits_{i=1}^{t_{\ell}}m_i}\middle/{m_0\prod\limits_{i\in\mathcal{B}^{(\ell)}}m_i}\right\rfloor + a_{s\ell}\right),\]
where $\mathcal{B}^{(\ell)}=\mathcal{B}\cap [N_{\ell}]$, and $a_{s\ell} \in \{0,1\}$ is a function of $s$ and $\ell$.
\end{lemma}

\begin{proof}
For any $s\in \mathcal{S}$, and $(g_1,g_2,\ldots,g_u)\in \Psi_1^{-1}(s)$, it holds that
       \begin{displaymath}\label{Eq: congruence equations}
         \left\{\begin{aligned}
            g_{1}\equiv &g_{2}\equiv\cdots\equiv g_{u}\equiv s\pmod{m_0},\\
            g_{1}\equiv &s_i^{(1)}\pmod{m_i}~\mathrm{for~all}~i\in \mathcal{B}\cap [N_1],\\
            g_{2}\equiv &s_i^{(2)}\pmod{m_i}~\mathrm{for~all}~i\in \mathcal{B}\cap [N_2],\\
                     &\vdots\\
            g_{u-1}\equiv&s_i^{(u-1)}\pmod{m_i}~\mathrm{for~all}~i\in \mathcal{B}\cap [N_{u-1}],\\
            g_{u}\equiv &s_i^{(u)}\pmod{m_i}~\mathrm{for~all}~i\in \mathcal{B}\cap [N_u].\\
           \end{aligned} \right.
       \end{displaymath}
where
     \[s_i^{(\ell)}=\begin{cases}
           (h_{\ell}(s_i)+w_i^{(\ell)})\pmod{m_i},~\mathrm{if}~\ell\in [u-1], i\in \mathcal{B}^{(\ell)}\subseteq [N_{\ell}],\\
            (h_{u}(s_i)+w_i^{(u)}) \pmod{m_i},~\mathrm{if}~\ell=u, i\in \mathcal{B}^{(\ell)}\cap [N_{u-1}],\\
            s_i,\mathrm{if}~\ell=u, i\in \mathcal{B}^{(\ell)}\cap [N_{u-1}+1,N_u],\\
         \end{cases}\]
By Lemma \ref{le: CRT}, it holds that
    \[g_{\ell}\equiv\sum\limits_{i\in \mathcal{\widetilde{B}}^{(\ell)}}\lambda_{i,\mathcal{\widetilde{B}}^{(\ell)}}M_{i,\mathcal{\widetilde{B}}^{(\ell)}}s_i^{(\ell)}\pmod {M_{\mathcal{\widetilde{B}}^{(\ell)}}},\ell\in [u],\]
where $\mathcal{\widetilde{B}}^{(\ell)}=\{0\}\cup(\mathcal{B}\cap [N_{\ell}])=\{0\}\cup\mathcal{B}^{(\ell)}$, $s_0^{(\ell)}=s$,
$M_{\mathcal{\widetilde{B}}^{(\ell)}}=\prod\limits_{i\in\mathcal{\widetilde{B}}^{(\ell)}}m_i$, $M_{i,\mathcal{\widetilde{B}}^{(\ell)}}=M_{\mathcal{\widetilde{B}}^{(\ell)}}/m_i$, and $\lambda_{i,\mathcal{\widetilde{B}}^{(\ell)}}\equiv M_{i,\mathcal{\widetilde{B}}^{(\ell)}}^{-1}\pmod {m_i}$.

Denote by $g_{\mathcal{\widetilde{B}}^{(\ell)}}=\sum\limits_{i\in \mathcal{\widetilde{B}}^{(\ell)}}\lambda_{i,\mathcal{\widetilde{B}}^{(\ell)}}M_{i,\mathcal{\widetilde{B}}^{(\ell)}}s_i^{(\ell)}\pmod {M_{\mathcal{\widetilde{B}}^{(\ell)}}}$, then
    \[g_{\mathcal{\widetilde{B}}^{(\ell)}}<M_{\mathcal{\widetilde{B}}^{(\ell)}}=\prod\limits_{i\in \mathcal{\widetilde{B}}^{(\ell)}}m_i,\]
and
      \begin{equation}\label{Eq: preimage F}
     g_{\ell}\equiv g_{\mathcal{\widetilde{B}}^{(\ell)}}\pmod{M_{\mathcal{\widetilde{B}}^{(\ell)}}}.
     \end{equation}
That is, there exists a nonnegative integer  $k_{\mathcal{\widetilde{B}}^{(\ell)}}$ such that
     \[g_{\ell}=g_{\mathcal{\widetilde{B}}^{(\ell)}}+k_{\mathcal{\widetilde{B}}^{(\ell)}} M_{\mathcal{\widetilde{B}}^{(\ell)}}.\]

On the one hand, $g_{\ell}$ satisfies $0\leq g_{\ell} <\prod\limits_{i=1}^{t_{\ell}}m_i$ in expression \eqref{Eq: definiton F}. Since $0\leq g_{\mathcal{\widetilde{B}}^{(\ell)}}<M_{\mathcal{\widetilde{B}}^{(\ell)}}$, it holds that
    \[0\leq k_{\mathcal{\widetilde{B}}^{(\ell)}}<\frac{\prod\limits_{i=1}^{t_{\ell}}m_i-g_{\mathcal{\widetilde{B}}^{(\ell)}}}{M_{\mathcal{\widetilde{B}}^{(\ell)}}}.\]
Let $\eta_{\mathcal{\widetilde{B}}^{(\ell)}}$ be the number of all possible $k_{\mathcal{\widetilde{B}}^{(\ell)}}$, then
        \begin{displaymath}
         \begin{aligned}
         \eta_{\mathcal{\widetilde{B}}^{(\ell)}}
         &=\begin{cases}
            \left\lfloor\left.{\prod\limits_{i=1}^{t_{\ell}}m_i}\right/{M_{\mathcal{\widetilde{B}}^{(\ell)}}}\right\rfloor,
                        ~\mathrm{if}~g_{\mathcal{\widetilde{B}}^{(\ell)}}\geq \left(\prod\limits_{i=1}^{t_{\ell}}m_i\pmod {M_{\mathcal{\widetilde{B}}^{(\ell)}}}\right),\\
            \left(\left\lfloor\left.{\prod\limits_{i=1}^{t_{\ell}}m_i}\right/{M_{\mathcal{\widetilde{B}}^{(\ell)}}}\right\rfloor
                        +1\right),~\mathrm{if}~g_{\mathcal{\widetilde{B}}^{(\ell)}}< \left(\prod\limits_{i=1}^{t_{\ell}}m_i\pmod {M_{\mathcal{\widetilde{B}}^{(\ell)}}}\right).
            \end{cases}\\
         &=\begin{cases}
            \left\lfloor\left.{\prod\limits_{i=1}^{t_{\ell}}m_i}\right/{\prod\limits_{i\in\mathcal{\widetilde{B}}^{(\ell)}}m_i}\right\rfloor,
                        ~\mathrm{if}~g_{\mathcal{\widetilde{B}}^{(\ell)}}\geq \left(\prod\limits_{i=1}^{t_{\ell}}m_i\pmod {M_{\mathcal{\widetilde{B}}^{(\ell)}}}\right),\\
            \left(\left\lfloor\left.{\prod\limits_{i=1}^{t_{\ell}}m_i}\right/{\prod\limits_{i\in\mathcal{\widetilde{B}}^{(\ell)}}m_i}\right\rfloor
                        +1\right),~\mathrm{if}~g_{\mathcal{\widetilde{B}}^{(\ell)}}< \left(\prod\limits_{i=1}^{t_{\ell}}m_i\pmod {M_{\mathcal{\widetilde{B}}^{(\ell)}}}\right).
            \end{cases}\\
         &=\left\lfloor\left.{\prod\limits_{i=1}^{t_{\ell}}m_i}\right/{m_0\prod\limits_{i\in\mathcal{B}^{(\ell)}}m_i}\right\rfloor+a_{s\ell},
                  \end{aligned}
       \end{displaymath}
where
  \begin{displaymath}
        a_{s\ell}=\begin{cases}
            0,~\mathrm{if}~g_{\mathcal{\widetilde{B}}^{(\ell)}}\geq \left(\prod\limits_{i=1}^{t_{\ell}}m_i\pmod {M_{\mathcal{\widetilde{B}}^{(\ell)}}}\right),\\
            1,~\mathrm{if}~g_{\mathcal{\widetilde{B}}^{(\ell)}}< \left(\prod\limits_{i=1}^{t_{\ell}}m_i\pmod {M_{\mathcal{\widetilde{B}}^{(\ell)}}}\right),
            \end{cases}\\
       \end{displaymath}
is a function of $s$ and $\ell$. On the other hand, different choices for nonnegative integer vector $(k_{\mathcal{\widetilde{B}}^{(1)}},k_{\mathcal{\widetilde{B}}^{(2)}}\ldots,k_{\mathcal{\widetilde{B}}^{(u)}})$ correspond to different $(g_1,g_2,\ldots,g_u)$ satisfying $0\leq g_{\ell}<\prod\limits_{i=1}^{t_{\ell}}m_i, \ell\in [u]$ and expression \eqref{Eq: preimage F}, i.e., $(g_1,g_2,\ldots,g_u)\in \mathcal{G}_1$. Therefore, \[|\Psi_1^{-1}(s)|=\prod_{\ell=1}^u \left(\left\lfloor\left.{\prod\limits_{i=1}^{t_{\ell}}m_i}\right/{m_0\prod\limits_{i\in\mathcal{B}^{(\ell)}}m_i}\right\rfloor + a_{s\ell}\right),\]
where $\mathcal{B}^{(\ell)}=\mathcal{B}\cap [N_{\ell}]$, and $a_{s\ell} \in \{0,1\}$ is a function of $s$ and $\ell$.
\end{proof}

\begin{theorem}\label{th: cardinality of G}
Denote by $\Omega_q=|\Psi_1^{-1}(s)|$, where $q=a_{s1}+2a_{s2}+2^2a_{s3}+\cdots+2^{u-1}a_{su}\in \mathbb{Z}_{2^u}$. Let $\delta_q$ be the size of the set
       \[\{s\in \mathcal{S}:|\Psi_1^{-1}(s)|=\Omega_q\}.\]
then $\delta_0,\delta_1,\ldots,\delta_{2^{u}-1}$ are non-negative integers, $\delta_0+\delta_1+\cdots+\delta_{2^{u}-1}=m_0$, and the size of $\mathcal{G}_1$ can be written as
 \[|\mathcal{G}_1|=\delta_0\Omega_0+\delta_1\Omega_1+\delta_2\Omega_2+\cdots+\delta_{2^{u}-1}\Omega_{2^{u}-1}.\]
\end{theorem}

\begin{proof}
For $s\in \mathcal{S}$ and $\ell\in [u]$, it has that $a_{s\ell}\in \{0,1\}$, and $q=a_{s1}+2a_{s2}+2^2a_{s3}+\cdots+2^{u-1}a_{su}\in \mathbb{Z}_{2^u}$. According to Lemma \ref{le: preimage}, the value $q$ is unique if the secret $s$ is given. As a result, $\delta_0+\delta_1+\cdots+\delta_{2^{u}-1}=|\mathcal{S}|=m_0$, and
                \[|\mathcal{G}_1|=\sum_{s\in \mathcal{S}}|\Psi_1^{-1}(s)|=\sum_{q\in \mathbb{Z}_{2^u}}\delta_q\Omega_q=\delta_0\Omega_0+\delta_1\Omega_1+\delta_2\Omega_2+\cdots+\delta_{2^{u}-1}\Omega_{2^{u}-1}.\]
\end{proof}

\begin{lemma}\label{le: The limits}
For all $\ell \in [u]$, it holds that
            \[\lim\limits_{m_0\to\infty}\frac{\prod\limits_{i=1}^{t_{\ell}}m_i}{m_0\prod\limits_{i\in\mathcal{B}^{(\ell)}}m_i}=k.\]
Moreover, there is an integer $b_1\in \mathbb{Z}_{2^u}$ such that
     \[\lim\limits_{m_0\to\infty}\frac{\delta_{b_1}}{m_0}=1, \lim\limits_{m_0\to\infty}\frac{{\sum\limits_{q\in \mathbb{Z}_{2^u},q\neq {b_1}}}\delta_q}{m_0}=0.\]
\end{lemma}

\begin{proof}
According to Definition \ref{def:k-compact sequences} of the $k$-compact sequence of co-primes, it is easy to check that for any $i,j\in [n]$,
                           \[\lim\limits_{m_0\to\infty}\frac{m_{i}}{m_{0}}=k,~\mathrm{and}~\lim\limits_{m_0\to\infty}\frac{m_{i}}{m_{j}}=1.\]
Recall that $|\mathcal{B}\cap(\cup_{w=1}^{\ell}\mathcal{P}_{w})|=|\mathcal{B}^{(\ell)}|=t_{\ell}-1~\mathrm{for~all}~\ell\in [u]$. Let $\mathcal{B}^{(\ell)}=\{i_1,i_2,\ldots,i_{t_{\ell}}-1\}$, then
     \[\lim\limits_{m_0\to\infty}\frac{\prod\limits_{i=1}^{t_{\ell}}m_i}{m_0\prod\limits_{i\in\mathcal{B}^{(\ell)}}m_i}
       =\lim\limits_{m_0\to\infty}\frac{m_1}{m_0} \times \lim\limits_{m_0\to\infty}\frac{m_2}{m_{i_{1}}} \times \cdots \times \lim\limits_{m_0\to\infty}\frac{m_{t_{\ell}}}{m_{i_{t_{\ell}-1}}}=k.\]

Since $L=\{m_0,m_1,m_2,\ldots,m_n\}$ is a $k$-compact sequence of co-primes, then its subsequences $L_{\ell}=\{m_0,m_1,m_2,\ldots,m_{N_{\ell}}\}, \ell\in [u]$ are also $k$-compact.    According to the item (3) in Lemma \ref{le: resutls-Tiplea2018}, there exist constants $b_{11},b_{12},\ldots,b_{1u}\in\mathbb{Z}_{2}$ and $u$ sets
          \begin{displaymath}
             \begin{aligned}
             \mathcal{D}_{\ell}=\{&q\in \mathbb{Z}_{2^{u}}: q=a_{q1}+2a_{q2}+2^2a_{q3}+\cdots+2^{u-1}a_{qu},~\mathrm{and}\\
                                   &a_{q\ell}=b_{1\ell}~\mathrm{is~constant}\}, \ell\in [u],
             \end{aligned}
          \end{displaymath}
such that
                            \[\lim\limits_{m_0\to\infty}\frac{\sum\limits_{q\in \mathcal{D}_{\ell}}\delta_q}{m_0}=1, \lim\limits_{m_0\to\infty}\frac{{\sum\limits_{q\in \mathbb{Z}_{2^u},q\notin \mathcal{D}_\ell}}\delta_q}{m_0}=0, \ell\in [u].\]
If $q\in \mathbb{Z}_{2^u}$ and $q\notin \mathcal{D}_{\ell}$, then $0\leq\frac{\delta_q}{m_0}\leq\frac{{\sum\limits_{q\in \mathbb{Z}_{2^u},q\notin \mathcal{D}_\ell}}\delta_q}{m_0}$. Consequently,
                         \begin{equation}\label{eq:qnumber/m0=0}
                         \lim\limits_{m_0\to\infty}\frac{\delta_q}{m_0}=0,~\mathrm{for~some}~\ell\in [u]~\mathrm{such~that}~q\notin \mathcal{D}_{\ell}.
                         \end{equation}

Denote by $b_1=b_{11}+2b_{12}+2^2b_{13}+\cdots+2^{u-1}b_{1u}$, then $\bigcap\limits_{\ell\in [u]}\mathcal{D}_{\ell}=\{b_1\}$. This shows that for any $q\in \mathbb{Z}_{2^u}$ and $q\neq b_1$, there exists $\ell_2\in [u]$ such that $q\notin \mathcal{D}_{\ell_2}$. By expression \eqref{eq:qnumber/m0=0}, we have
                      \[\lim\limits_{m_0\to\infty}\frac{\delta_q}{m_0}=0,~\mathrm{for~all}~q\in \mathbb{Z}_{2^u},q\neq b_1.\]
Therefore,
                  \[\lim\limits_{m_0\to\infty}\frac{{\sum\limits_{q\in \mathbb{Z}_{2^u},q\neq b_1}}\delta_q}{m_0}=\sum\limits_{q\in \mathbb{Z}_{2^u},q\neq b_1}\lim\limits_{m_0\to\infty}\frac{\delta_q}{m_0}=0.\]
Since $\delta_0+\delta_1+\cdots+\delta_{2^{u}-1}=m_0$ from Theorem \ref{th: cardinality of G}, then
                      \[\lim\limits_{m_0\to\infty}\frac{\delta_{b_1}}{m_0}=\lim\limits_{m_0\to\infty}\frac{m_0-{\sum\limits_{q\in \mathbb{Z}_{2^u},q\neq b_1}}\delta_q}{m_0}=1.\]
\end{proof}

\begin{lemma}\label{le:bound}
Suppose that there is an $s\in \mathcal{S}$ such that $\Omega_q=|\Psi_1^{-1}(s)|$, then for a
sufficiently large $m_0$, the following results hold.
\begin{itemize}
  \item[(1)] $1\leq\Omega_q\leq \prod_{\ell=1}^u                                                    \left(\left\lfloor{\prod\limits_{i=1}^{t_{\ell}}m_i}\middle/{m_0\prod\limits_{i\in\mathcal{B}^{(\ell)}}m_i}\right\rfloor +1\right)\leq (k+1)^u$.
  \item[(2)] $1\leq \frac{\Omega_q}{\Omega_{b_1}}\leq (k+1)^u$.
  \item[(3)] $\lim\limits_{m_0\to\infty}\frac{m_0\Omega_{b_1}}{|\mathcal{G}_1|}=1$.
\end{itemize}
\end{lemma}

\begin{proof}
(1) From Lemma \ref{le: The limits}, it holds that $\lim\limits_{m_0\to\infty}\frac{\prod\limits_{i=1}^{t_{\ell}}m_i}{m_0\prod\limits_{i\in\mathcal{B}^{(\ell)}}m_i}=k, \ell \in [u]$.
These show that for a sufficiently large $m_0$, we have
               \[k-1<\frac{\prod\limits_{i=1}^{t_{\ell}}m_i}{m_0\prod\limits_{i\in\mathcal{B}^{(\ell)}}m_i}<k+1, \ell\in [u].\]
Consequently, $\prod_{\ell=1}^u \left(\left\lfloor{\prod\limits_{i=1}^{t_{\ell}}m_i}\middle/{m_0\prod\limits_{i\in\mathcal{B}^{(\ell)}}m_i}\right\rfloor +1\right)\leq (k+1)^u$.
Since $\Omega_q=|\Psi^{-1}(s)|$ for some $s\in \mathcal{S}$, then  by Lemma \ref{le: preimage}, it holds that
                          \[1\leq\Omega_q\leq \prod_{\ell=1}^u                                                    \left(\left\lfloor{\prod\limits_{i=1}^{t_{\ell}}m_i}\middle/{m_0\prod\limits_{i\in\mathcal{B}^{(\ell)}}m_i}\right\rfloor +1\right)\leq (k+1)^u.\]

(2) From the item (1) in this lemma, we have
       \[1\leq\Omega_{b_1}\leq \prod_{\ell=1}^u                                                    \left(\left\lfloor{\prod\limits_{i=1}^{t_{\ell}}m_i}\middle/{m_0\prod\limits_{i\in\mathcal{B}^{(\ell)}}m_i}\right\rfloor +1\right)\leq (k+1)^u.\]
As a result, $1\leq \frac{\Omega_q}{\Omega_{b_1}}\leq (k+1)^u$ for a sufficiently large $m_0$.

(3) By the item (1) and item (2) of this lemma, for a sufficiently large $m_0$ we have
    \[\frac{m_0}{\delta_{b_1}+(k+1)^u\sum\limits_{q\in \mathbb{Z}_{2^u},q\neq {b_1}}\delta_q}\leq\frac{m_0\Omega_{b_1}}{|\mathcal{G}_1|}=\log_2\frac{m_0\Omega_{b_1}}{\sum\limits_{q\in \mathbb{Z}_{2^u}}\delta_q\Omega_q}\leq\frac{m_0\Omega_{b_1}}{\delta_{b_1}\Omega_{b_1}}\]
According to Lemma \ref{le: The limits}, we have $\lim\limits_{m_0\to\infty}\frac{\delta_{b_1}}{m_0}=1, \lim\limits_{m_0\to\infty}\frac{{\sum\limits_{q\in \mathbb{Z}_{2^u},q\neq {b_1}}}\delta_q}{m_0}=0$. Therefore,
        \[\lim\limits_{m_0\to\infty}\frac{m_0}{\delta_{b_1}+(k+1)^u\sum\limits_{q\in \mathbb{Z}_{2^u},q\neq {b_1}}\delta_q}=1=\lim\limits_{m_0\to\infty}\frac{m_0\Omega_{b_1}}{\delta_{b_1}\Omega_{b_1}},\]
which means that $\lim\limits_{m_0\to\infty}\frac{m_0\Omega_{b_1}}{|\mathcal{G}_1|}=1$.
\end{proof}

\begin{theorem}[Asymptotic perfectness]\label{Th: our perfectness1}
   Our scheme is asymptotically perfect.
\end{theorem}

\begin{proof}
 The secret in our scheme is randomly and uniformly, which means that $\mathsf{H}(\textbf{S})=\log_2|\mathcal{S}|=\log_2 m_0$. For any $s\in \mathcal{S}$, there is a unique $q\in \mathbb{Z}_{2^u}$ such that $|\Psi_1^{-1}(s)|=\Omega_q$ by Theorem \ref{th: cardinality of G}. Therefore, the loss entropy
 \begin{displaymath}
             \begin{aligned}
              &\Delta(\textbf{V}^\prime_\mathcal{B}) = \mathsf{H}(\mathbf{S}) - \mathsf{H}(\mathbf{S}|\textbf{V}^\prime_\mathcal{B})\\
                    =&\log_2 m_0-
                      \sum_{s\in \mathcal{S}}
                      \mathsf{Pr}(\mathbf{S}=s|\textbf{V}^\prime_\mathcal{B}=\mathcal{V}^\prime_{\mathcal{B}})\log_2 \frac{1}{\mathsf{Pr}(\mathbf{S}=s|\textbf{V}^\prime_\mathcal{B}=\mathcal{V}^\prime_{\mathcal{B}})}\\
                    =&\log_2 m_0+\sum_{s\in \mathcal{S}}\frac{|\Psi_1^{-1}(s)|}{|\mathcal{G}_1|}\log_2\frac{|\Psi^{-1}(s)|}{|\mathcal{G}_1|}\\
                    =&\log_2 m_0+\sum_{q\in \mathbb{Z}_{2}}\delta_q\frac{\Omega_q}{|\mathcal{G}_1|}\log_2\frac{\Omega_q}{|\mathcal{G}_1|}\\
                    =&\log_2 m_0+\sum_{q\in \mathbb{Z}_{2^u}}\delta_q\frac{\Omega_q}{|\mathcal{G}_1|}\log_2\frac{\Omega_{b_1}}{|\mathcal{G}_1|}
                    +\sum_{q\in \mathbb{Z}_{2^u},q\neq {b_1}}\delta_q\frac{\Omega_q}{|\mathcal{G}_1|}\log_2\frac{\Omega_q}{\Omega_{b_1}}\\
                    =&\log_2 m_0+\log_2\frac{\Omega_{b_1}}{|\mathcal{G}_1|}
                    +\sum_{q\in \mathbb{Z}_{2^u},q\neq {b_1}}\delta_q\frac{\Omega_q}{|\mathcal{G}_1|}\log_2\frac{\Omega_q}{\Omega_{b_1}}\\
                    =&\log_2\frac{m_0\Omega_{b_1}}{|\mathcal{G}_1|}+\sum_{q\in \mathbb{Z}_{2^u},q\neq {b_1}}\delta_q\frac{\Omega_q}{|\mathcal{G}_1|}\log_2\frac{\Omega_q}{\Omega_{b_1}}.
            \end{aligned}
          \end{displaymath}
By Lemma \ref{le:bound}, for a sufficiently large $m_0$ we have
    \[0\leq\Delta(\textbf{V}^\prime_\mathcal{B})\leq \log_2\frac{m_0\Omega_{b_1}}{|\mathcal{G}_1|}+\sum_{q\in \mathbb{Z}_{2^u},q\neq {b_1}}\delta_q\frac{\Omega_q}{|\mathcal{G}_1|}\log_2(k+1)^u.\]
According to Lemmas \ref{le: The limits} and \ref{le:bound}, we get that
     \begin{displaymath}
             \begin{aligned}
             &\lim\limits_{m_0\to\infty}\left(\log_2\frac{m_0\Omega_{b_1}}{|\mathcal{G}_1|}+\sum_{q\in \mathbb{Z}_{2^u},q\neq {b_1}}\delta_q\frac{\Omega_q}{|\mathcal{G}_1|}\log_2(k+1)^u\right)\\
           =&\lim\limits_{m_0\to\infty}\left(\log_2\frac{m_0\Omega_{b_1}}{|\mathcal{G}_1|}+\frac{|\mathcal{G}_1|-\delta_{b_1}\Omega_{b_1}}{|\mathcal{G}_1|}\log_2(k+1)^u\right)\\
           =&\lim\limits_{m_0\to\infty}\left(\log_2\frac{m_0\Omega_{b_1}}{|\mathcal{G}_1|}+\left(1-\frac{\delta_{b_1}}{m_0}\cdot\frac{m_0\Omega_{b_1}}{|\mathcal{G}_1|}\right)\log_2(k+1)^u\right)\\
           =&0,
     \end{aligned}
          \end{displaymath}
which means that $\lim\limits_{m_0\to\infty}\Delta(\textbf{V}^\prime_\mathcal{B})$=0. Therefore, for all $\epsilon_1>0$, there is a positive integer $\sigma_{11}$ such that for all $|\mathcal{S}|=m_0>\sigma_{11}$, it holds that
     \[0<\mathsf{H}(\mathbf{S})-\mathsf{H}(\mathbf{S}|\mathbf{V}^{\prime}_{\mathcal{B}})<\frac{1}{2}\epsilon_1.\]

By Lemma \ref{le: loss entropy of DHSS}, for the above $\epsilon_{1}$, there is a positive integer $\sigma_{12}$ such that for all $|\mathcal{S}|=m_0>\sigma_{12}$, it holds that
                               \[0<\mathsf{H}(\mathbf{S}|\mathbf{V}'_{\mathcal{B}})-\mathsf{H}(\mathbf{S}|\mathbf{V}_{\mathcal{B}})<\frac{1}{2}\epsilon_1.\]
As a result, there is a positive integer $\sigma_{1}=\max\{\sigma_{11},\sigma_{12}\}$ such that for all $|\mathcal{S}|=m_0>\sigma_{1}$, we have
     \[0<\mathsf{H}(\mathbf{S})-\mathsf{H}(\mathbf{S}|\mathbf{V}_{\mathcal{B}})=(\mathsf{H}(\mathbf{S})-\mathsf{H}(\mathbf{S}|\mathbf{V}^{\prime}_{\mathcal{B}}))
                               +(\mathsf{H}(\mathbf{S}|\mathbf{V}'_{\mathcal{B}})-\mathsf{H}(\mathbf{S}|\mathbf{V}_{\mathcal{B}})<\epsilon_1.\]
Therefore, our scheme is asymptotically perfect by Definition \ref{def: Asymptotically Ideal HSS}.
\end{proof}

\begin{theorem}\label{Theorem: summary1}
Let $\mathcal{P}$ be a set of $n$ participants, and it is partitioned into $u$ disjoint subsets $\mathcal{P}_1, \mathcal{P}_2,\ldots, \mathcal{P}_u$. For a threshold sequence $t_1, t_2, \ldots, t_u$ satisfying $1\leq t_1< t_2<\cdots<t_u\leq n$ and $t_{\ell}\leq |\mathcal{P}_{\ell}|$ for $\ell\in [u]$, our scheme 1 is a secure and asymptotically perfect DHSS scheme. Moreover, if $L=\{m_{0},m_{1},\ldots,m_{n}\}$ is a $1$-compact sequence of co-primes, then our scheme 1 is an asymptotically ideal DHSS scheme.
\end{theorem}
\begin{proof}
By Definition \ref{def: Multilevel secret sharing of disjunctive}, Theorem \ref{Theorem:Correctness of our scheme1} and Theorem \ref{Th: our perfectness1}, our scheme 1 is a secure and asymptotically perfect DHSS scheme. Recall that the shares are given by
         \[s_i=\begin{cases}
            c_i,\mathrm{if}~i\in [N_{u-1}],\\
            y_u \pmod {m_i},\mathrm{if}~i\in [N_{u-1}+1,N_u],\\
         \end{cases}\]
where $c_{i}\in \mathbb{Z}_{m_{i}}$ for $i\in [N_{u-1}]$. When $L=\{m_{0},m_{1},\ldots,m_{n}\}$ is a $1$-compact sequence of co-primes, there exists $\theta \in (0,1)$ such that
                               \[m_{0}<m_{i}<m_{0}+m_{0}^{\theta}.\]
The secret space is $\mathcal{S}=\mathbb{Z}_{m_0}$ and the maximum share space is $\mathbb{Z}_{m_n}$. Consequently, the information $\rho_1$ of our scheme 1 satisfies
                  \[\frac{\log_2 m_{0}}{\log_2 (m_{0}+m_{0}^{\theta})}<\rho_1=\frac{\log_2 m_{0}}{\log_2 m_{n}}<\frac{\log_2 m_{0}}{\log_2 m_{0}}=1.\]
Since $\lim\limits_{m_0\to\infty}\frac{\log_2 m_{0}}{\log_2 (m_{0}+m_{0}^{\theta})}=1$, we have $\lim\limits_{m_0\to\infty}\rho_1=1$. Hence, by Definition \ref{def: Asymptotically Ideal HSS}, our scheme 1 is an asymptotically ideal DHSS scheme.
\end{proof}

\section{A novel asymptotically ideal CHSS scheme}\label{Sec: our scheme}
In this section, we construct a hierarchical secret sharing scheme for conjunctive access structures. We present the scheme in Subsection~\ref{Subsec:our CHSS} and analyze its security in Subsection~\ref{Subsec: security of our CHSS}.

\subsection{Our scheme 2}\label{Subsec:our CHSS}
Let $\mathcal{P}$ be a set of $n$ participants, divided into $u$ pairwise disjoint subsets $\mathcal{P}_1, \mathcal{P}_2, \dots, \mathcal{P}_u$. For each $\ell\in[u]$, define $n_{\ell}=|\mathcal{P}_{\ell}|$ and let $N_{\ell}=\sum_{w=1}^{\ell} n_{w}$. Additionally, for every $\ell\in[u]$, denote by $t_{\ell}$ the threshold assigned to the union $\bigcup_{w=1}^{\ell}\mathcal{P}_w$. The thresholds satisfy $1\leq t_1<t_2<\cdots<t_u$ and $t_{\ell}\leq N_{\ell}$ for all $\ell\in [u]$. Our scheme 2 consists of two phases: the share generation phase and the secret reconstruction phase.

\textbf{(1) Share Generation Phase}. Let $m_{0}$ be a big positive integer, and let $\mathcal{S}=\mathbb{Z}_{m_{0}}$ be the secret space.
 \begin{itemize}
   \item[Step 1.] The dealer generates and publishes a $k$-compact sequence of co-primes $L=\{m_{0},m_{1},\\
       \ldots,m_{n}\}$, namely, the following three conditions are satisfied.
          \begin{itemize}
            \item[(i)] $m_0, m_1,\ldots, m_n$ are pairwise co-primes.
             \item[(ii)] $m_{0}<m_{1}<\cdots<m_{n}$.
             \item[(iii)]$km_{0}<m_{i}<km_{0}+m_{0}^{\theta}$ for all $i \in [n]$, where $k \geq 1$ and $\theta \in (0,1)$ are real numbers.
           \end{itemize}

   \item[Step 2.] For any given secret $s \in \mathcal{S}$, the dealer selects random integers $c_{i} \in \mathbb{Z}_{m_{i}}$, $i \in [N_{u-1}]$, and random integers
           $r_{\ell}\in \mathcal{S}$ and $\alpha_{\ell}$ for $\ell\in [u]$ such that $s=\sum\limits_{\ell\in [u]}r_{\ell} \pmod {m_0}$, and
           \[0 \leq y_{\ell}=s+\alpha_{\ell}m_{0} < \prod\limits_{i=1}^{t_{\ell}}m_i, \ell\in [u].\]

           The dealer sends the share $s_i$ to the $i$-th participant, where
                \[s_i=\begin{cases}
                   c_i,\mathrm{if}~i\in [N_{u-1}],\\
                 y_u \pmod {m_i},\mathrm{if}~i\in [N_{u-1}+1,N_u].\\
                \end{cases}\]

       \item[Step 3.] The dealer selects $u$ publicly known distinct one-way functions $h_{\ell}, \ell\in [u]$. Each function accepts an input of arbitrary length and produces an output of fixed length $\lfloor \log_{2} m_n \rfloor$. Then the dealer publishes the values
                    \[w_i^{(\ell)}=(y_{\ell}-h_{\ell}(s_i))\pmod{m_i}, \ell\in [u-1], i\in [N_{\ell}],\]
           and $w_i^{(u)}=(y_{u}-h_{u}(s_i))\pmod{m_i}$, $i\in [N_{u-1}]$.
   \end{itemize}

\textbf{(2) Secret Reconstruction Phase}. For any $\mathcal{A} \subseteq \mathcal{P}$ such that
        \[\mathcal{A}^{(\ell)}=\mathcal{A}\cap (\cup_{w=1}^{\ell}\mathcal{P}_{w}), |\mathcal{A}^{(\ell)}|\geq t_{\ell}~\mathrm{for~all}~\ell\in [u],\]
participants of $\mathcal{A}^{(\ell)}$ compute

                          \[y_{\ell}=\sum\limits_{i\in \mathcal{A}^{(\ell)}}\lambda_{i,\mathcal{A}^{(\ell)}}M_{i,\mathcal{A}^{(\ell)}}s_i^{(\ell)} \pmod {M_{\mathcal{A}^{(\ell)}}},\]
where $M_{\mathcal{A}^{(\ell)}}=\prod\limits_{i\in\mathcal{A}^{(\ell)}}m_i$, $M_{i,\mathcal{A}^{(\ell)}}=M_{\mathcal{A}^{(\ell)}}/m_i$, $\lambda_{i,\mathcal{A}^{(\ell)}}\equiv M_{i,\mathcal{A}^{(\ell)}}^{-1}\pmod {m_i}$, and
          \[s_i^{(\ell)}=\begin{cases}
           (h_{\ell}(s_i)+w_i^{(\ell)})\pmod{m_i},~\mathrm{if}~\ell\in [u-1], i\in \mathcal{A}^{(\ell)}\subseteq [N_{\ell}],\\
            (h_{u}(s_i)+w_i^{(u)}) \pmod{m_i},~\mathrm{if}~\ell=u, i\in \mathcal{A}^{(\ell)}\cap [N_{u-1}],\\
            s_i,\mathrm{if}~\ell=u, i\in \mathcal{A}^{(\ell)}\cap [N_{u-1}+1,N_u].\\
         \end{cases}\]
As a result, they obtain $r_{\ell}=y_{\ell}\pmod{m_0}, \ell\in [u]$, and reconstruct the secret as
         \[s=\sum\limits_{\ell\in [u]}r_{\ell} \pmod {m_0}.\]

\subsection{Security analysis of our scheme}\label{Subsec: security of our CHSS}
We will prove the correctness, asymptotic perfectness and asymptotic ideal of our scheme in this subsection.
\begin{theorem}[Correctness]\label{Theorem:Correctness of our scheme2} Any subset
       \[\mathcal{A}\in \Gamma_2=\{\mathcal{A} \subseteq \mathcal{P}: |\mathcal{A}\cap(\bigcup_{w=1}^{\ell}\mathcal{P}_{w})|\geq t_{\ell} ~for~ \forall \ell \in [u]\}\]
       can reconstruct the secret.
\end{theorem}

\begin{proof}
Let $\mathcal{A}\in \Gamma_2$, then $\mathcal{A}\subseteq \mathcal{P}, \mathcal{A}^{(\ell)}=\mathcal{A}\cap (\cup_{w=1}^{\ell}\mathcal{P}_{w})$, and $|\mathcal{A}^{(\ell)}|\geq t_{\ell}$ for all $\ell\in [u]$. Since $w_i^{(\ell)}=(y_{\ell}-h_{\ell}(s_i))\pmod{m_i}, \ell\in [u-1], i\in [N_{\ell}]$,
and $w_i^{(u)}=(y_{u}-h_{u}(s_i))\pmod{m_i}$, $i\in [N_{u-1}]$ are public values, then for each $\ell\in [u]$, participants of $\mathcal{A}^{(\ell)}$ compute
       \[s_i^{(\ell)}=\begin{cases}
           (h_{\ell}(s_i)+w_i^{(\ell)})\pmod{m_i},~\mathrm{if}~\ell\in [u-1], i\in \mathcal{A}^{(\ell)}\subseteq [N_{\ell}],\\
            (h_{u}(s_i)+w_i^{(u)}) \pmod{m_i},~\mathrm{if}~\ell=u, i\in \mathcal{A}^{(\ell)}\cap [N_{u-1}],\\
            s_i,\mathrm{if}~\ell=u, i\in \mathcal{A}^{(\ell)}\cap [N_{u-1}+1,N_u],\\
         \end{cases}\]
and obtain the system of congruences
      \[y_{\ell} \equiv s_{i}^{(\ell)} \pmod{m_{i}}, i\in \mathcal{A}^{(\ell)}.\]
Recall that $m_{0}<m_{1}<\cdots<m_{n}$ and $y_{\ell}<\prod\limits_{i=1}^{t_{\ell}}m_i$ then $y_{\ell}<\prod\limits_{i\in \mathcal{A}^{(\ell)}}m_i$. According to the CRT for integer ring in Lemma \ref{le: CRT}, it holds that
   \[y_{\ell}=\sum\limits_{i\in \mathcal{A}^{(\ell)}}\lambda_{i,\mathcal{A}^{(\ell)}}M_{i,\mathcal{A}^{(\ell)}}s_i^{(\ell)} \pmod {M_{\mathcal{A}^{(\ell)}}},~\mathrm{and}~ r_{\ell}=y_{\ell}\pmod{m_0},\]
where $M_{\mathcal{A}^{(\ell)}}=\prod\limits_{i\in\mathcal{A}^{(\ell)}}m_i$, $M_{i,\mathcal{A}^{(\ell)}}=M_{\mathcal{A}^{(\ell)}}/m_i$ and $\lambda_{i,\mathcal{A}^{(\ell)}} \equiv M_{i,\mathcal{A}^{(\ell)}}^{-1}\pmod {m_i}$. Finally, the secret $s$ is constructed as $s=\sum\limits_{\ell\in [u]}r_{\ell}\pmod{m_0}$.
\end{proof}

Now we prove the \emph{asymptotic perfectness} of our scheme 2. Let $\mathcal{B}\subset \mathcal{P}$, and
            \[\mathcal{B}\notin \Gamma_2=\{\mathcal{A} \subseteq \mathcal{P}: |\mathcal{A}\cap(\bigcup_{w=1}^{\ell}\mathcal{P}_{w})|\geq t_{\ell} ~for~ \forall \ell \in [u]\}.\]
Consider the worst case, i.e.,
    \[|\mathcal{B}\cap(\cup_{w=1}^{\ell}\mathcal{P}_{w})|\geq t_{\ell}~\mathrm{for~all}~\ell\in [u-1]~~
    and~~ |\mathcal{B}\cap(\cup_{w=1}^{u}\mathcal{P}_{w})|= t_{u}-1.\]
From the proof of Theorem \ref{Theorem:Correctness of our scheme2}, participants of $\mathcal{B}$ can determine
              \[y_1,y_2,\ldots,y_{u-1}, r_1,r_2,\ldots,r_{u-1}.\]
Moreover, they have access to their own shares, the upper and lower bounds of $y_{u}$, as well as all publicly available information. Consequently, they attempt to reconstruct the secret by first choosing an integer $g_u\in \mathbb{Z}$ that meets the following five conditions, and then calculating $\left(g_u+\sum\limits_{\ell\in [u-1]}r_{\ell}\right) \pmod{m_0}$.

\begin{itemize}
  \item[(I)] $0\leq g_{u}<\prod\limits_{i=1}^{t_{u}}m_i$.
  \item[(II)] $g_{u}\equiv r_{u} \equiv \left(s-\sum\limits_{\ell\in [u-1]}r_{\ell}\right) \pmod{m_0}$.
  \item[(III)] $g_{u}\equiv(w_i^{(u)}+h_{u}(s_i))\pmod{m_i}, i\in \mathcal{B}\cap[N_{u-1}]$.
  \item[(IV)] $g_{u}\equiv s_i\pmod{m_i}$, $i\in \mathcal{B}\cap [N_{u-1}+1, N_u]$.
  \item[(V)] For any $i\in [N_{u-1}]$ with $i\notin \mathcal{B}$, suppose $i$ belongs to level $\mathcal{P}_{\ell_3}$. Then there exists an integer $\widetilde{s}_i\in \mathbb{Z}_{m_i}$ such that for all $\ell\in [\ell_3,u]$, the relation $h_{\ell}(\widetilde{s}_i)\equiv(g_{\ell}-w_i^{(\ell)})\pmod{m_i}$ holds, where the polynomials $g_{\ell}=y_{\ell}, \ell\in [u-1]$ are known.
  \end{itemize}

By the same discussion as in Lemma \ref{le: loss entropy of DHSS}, we obtain the following Lemma \ref{le: loss entropy of CHSS}.
\begin{lemma}\label{le: loss entropy of CHSS}
Let $\mathcal{V}_{\mathcal{B}}$ denote the set of conditions (i) through (v) that are the knowledge of $\mathcal{B}$. Denote by $\mathcal{V}^\prime_{\mathcal{B}}$ the set consisting of conditions (i) through (iv). For any $\epsilon_1>0$, there exists a positive integer $\sigma_1$ such that whenever $|\mathcal{S}|=m_0>\sigma_1$, the following holds:
                               \[0<\mathsf{H}(\mathbf{S}|\mathbf{V}^\prime_{\mathcal{B}})-\mathsf{H}(\mathbf{S}|\mathbf{V}_{\mathcal{B}})<\epsilon_1,\]
where $\mathbf{V}^\prime_{\mathcal{B}}$ and $\mathbf{V}_{\mathcal{B}}$ are random variables corresponding to $\mathcal{V}^\prime_{\mathcal{B}}$ and $\mathcal{V}_{\mathcal{B}}$, respectively.
\end{lemma}

Let us define
\begin{equation}\label{Eq: definiton }
\mathcal{G}_2=\{g_u\in \mathbb{Z}: ~\text{conditions (I)~through~(IV) hold}\}.
\end{equation}

\begin{lemma}\label{le: preimag}
Define the mapping $\Psi_2$ by
 \[\Psi_2: \mathcal{G}_2\mapsto \mathcal{S}, g_u \mapsto \left(g_u+\sum\limits_{\ell\in [u-1]}r_{\ell}\right) \pmod{m_0}.\]
 For any $s\in \mathcal{S}$, let
           \[\Psi_2^{-1}(s)=\left\{g_u\in \mathcal{G}_2: g_{u}\equiv\left(s-\sum\limits_{\ell\in [u-1]}r_{\ell}\right) \pmod{m_0}\right\},\]
then the size of the set $\Psi_2^{-1}(s)$ can be written as
            \[|\Psi_2^{-1}(s)|= \left\lfloor{\prod\limits_{i=1}^{t_{u}}m_i}\middle/{m_0\prod\limits_{i\in\mathcal{B}}m_i}\right\rfloor + a_{su},\]
where $a_{su} \in \{0,1\}$ is a function of $s$ and $u$.
\end{lemma}

\begin{proof}
For any $s\in \mathcal{S}$, and $g_u\in \Psi_2^{-1}(s)$, it holds that
       \begin{displaymath}\label{Eq: congruence equations}
         \left\{\begin{aligned}
            &g_{u}\equiv g_{u}\equiv\left(s-\sum\limits_{\ell\in [u-1]}r_{\ell}\right) \pmod{m_0},\\
            &g_{u}\equiv s_i^{(u)}\pmod{m_i}~\mathrm{for~all}~i\in \mathcal{B},\\
           \end{aligned} \right.
       \end{displaymath}
where
     \[s_i^{(u)}=\begin{cases}
            (h_{u}(s_i)+w_i^{(u)}) \pmod{m_i},~\mathrm{if}~i\in \mathcal{B}\cap [N_{u-1}],\\
            s_i,\mathrm{if}~ i\in \mathcal{B}\cap [N_{u-1}+1,N_u].\\
         \end{cases}\]
By the same discussion as in Lemma \ref{le: preimage}, we obtain
      \[|\Psi_2^{-1}(s)|= \left\lfloor\left.{\prod\limits_{i=1}^{t_{u}}m_i}\right/{m_0\prod\limits_{i\in\mathcal{B}}m_i}\right\rfloor + a_{su},\]
where $a_{su} \in \{0,1\}$ is a function of $s$ and $u$.
\end{proof}

\begin{theorem}\label{th: cardinality of }
Denote by $\Omega_q=|\Psi_2^{-1}(s)|$, where $q=a_{su}\in \mathbb{Z}_{2}$. Let $\delta_q$ be the size of the set
       \[\{s\in \mathcal{S}:|\Psi_2^{-1}(s)|=\Omega_q\}.\]
then $\delta_0,\delta_1$ are non-negative integers, $\delta_0+\delta_1=m_0$, and the size of $\mathcal{G}_2$ can be written as
 \[|\mathcal{G}_2|=\delta_0\Omega_0+\delta_1\Omega_1.\]
\end{theorem}

\begin{proof}
For $s\in \mathcal{S}$, it has that $a_{su}\in \{0,1\}$, and $q=a_{su}\in \mathbb{Z}_{2}$. According to Lemma \ref{le: preimag}, the value $q$ is unique if the secret $s$ is given. As a result, $\delta_0+\delta_1=|\mathcal{S}|=m_0$, and
                \[|\mathcal{G}_2|=\sum_{s\in \mathcal{S}}|\Psi_2^{-1}(s)|=\sum_{q\in \mathbb{Z}_{2}}\delta_q\Omega_q=\delta_0\Omega_0+\delta_1\Omega_1.\]
\end{proof}

As special cases of Lemmas \ref{le: The limits} and \ref{le:bound}, we get the following lemma.

\begin{lemma}\label{le: specail-case-sec3}
Suppose that there is an $s\in \mathcal{S}$ such that $\Omega_q=|\Psi_2^{-1}(s)|$, then for a
sufficiently large $m_0$, the following results hold.
\begin{itemize}
  \item[(1)] $\lim\limits_{m_0\to\infty}\frac{\prod\limits_{i=1}^{t_{u}}m_i}{m_0\prod\limits_{i\in\mathcal{B}^{(u)}}m_i}=k$.
  \item[(2)] There is an integer $b_2\in \mathbb{Z}_{2}$ such that
     \[\lim\limits_{m_0\to\infty}\frac{\delta_{b_2}}{m_0}=1, \lim\limits_{m_0\to\infty}\frac{{\sum\limits_{q\in \mathbb{Z}_{2},q\neq {b_2}}}\delta_q}{m_0}=0.\]
  \item[(3)] $1\leq\Omega_q\leq \left\lfloor{\prod\limits_{i=1}^{t_{\ell}}m_i}\middle/{m_0\prod\limits_{i\in\mathcal{B}^{(\ell)}}m_i}\right\rfloor +1\leq k+1$.
  \item[(4)] $1\leq \frac{\Omega_q}{\Omega_{b_2}}\leq k+1$.
  \item[(5)] $\lim\limits_{m_0\to\infty}\frac{m_0\Omega_{b_2}}{|\mathcal{G}_2|}=1$.
\end{itemize}
\end{lemma}

\begin{theorem}[Asymptotic perfectness]\label{Th: our perfectness2}
   Our scheme 2 is asymptotically perfect.
\end{theorem}

\begin{proof}
 The secret in our scheme 2 is randomly and uniformly, which means that $\mathsf{H}(\textbf{S})=\log_2|\mathcal{S}|=\log_2 m_0$. For any $s\in \mathcal{S}$, there is a unique $q\in \mathbb{Z}_{2}$ such that $|\Psi_2^{-1}(s)|=\Omega_q$ by Theorem \ref{th: cardinality of } . Therefore, the loss entropy
 \begin{displaymath}
                          \begin{aligned}
              &\Delta(\textbf{V}^\prime_\mathcal{B}) = \mathsf{H}(\mathbf{S}) - \mathsf{H}(\mathbf{S}|\textbf{V}^\prime_\mathcal{B})\\
                    =&\log_2 m_0-
                      \sum_{s\in \mathcal{S}}
                      \mathsf{Pr}(\mathbf{S}=s|\textbf{V}^\prime_\mathcal{B}=\mathcal{V}^\prime_{\mathcal{B}})\log_2 \frac{1}{\mathsf{Pr}(\mathbf{S}=s|\textbf{V}^\prime_\mathcal{B}=\mathcal{V}^\prime_{\mathcal{B}})}\\
                    =&\log_2 m_0+\sum_{s\in \mathcal{S}}\frac{|\Psi_{2}^{-1}(s)|}{|\mathcal{G}_{2}|}\log_2\frac{|\Psi_{2}^{-1}(s)|}{|\mathcal{G}_{2}|}\\
                    =&\log_2 m_0+\delta_1\frac{\Omega_1}{|\mathcal{G}_{2}|}\log_2\frac{\Omega_1}{|\mathcal{G}_{2}|}
                    +\delta_2\frac{\Omega_2}{|\mathcal{G}_{2}|}\log_2\frac{\Omega_2}{|\mathcal{G}_{2}|}\\
                    =&\log_2\frac{m_0\Omega_{b_2}}{|\mathcal{G}_2|}+\sum_{q\in \mathbb{Z}_{2},q\neq {b_2}}\delta_q\frac{\Omega_q}{|\mathcal{G}_2|}\log_2\frac{\Omega_q}{\Omega_{b_2}}.
            \end{aligned}
          \end{displaymath}
By Lemma \ref{le: specail-case-sec3}, for a sufficiently large $m_0$ we have
    \[0\leq\Delta(\textbf{V}^\prime_\mathcal{B})\leq \log_2\frac{m_0\Omega_{b_2}}{|\mathcal{G}_2|}+\sum_{q\in \mathbb{Z}_{2^u},q\neq {b_2}}\delta_q\frac{\Omega_q}{|\mathcal{G}_2|}\log_2(k+1)
    .\]
Since
     \begin{displaymath}
             \begin{aligned}
             &\lim\limits_{m_0\to\infty}\left(\log_2\frac{m_0\Omega_{b_2}}{|\mathcal{G}_2|}+\sum_{q\in \mathbb{Z}_{2^u},q\neq {b_2}}\delta_q\frac{\Omega_q}{|\mathcal{G}_2|}\log_2(k+1)\right)\\
           =&\lim\limits_{m_0\to\infty}\left(\log_2\frac{m_0\Omega_{b_2}}{|\mathcal{G}_2|}+\frac{|\mathcal{G}_2|-\delta_{b_2}\Omega_{b_2}}{|\mathcal{G}_2|}\log_2(k+1)\right)\\
           =&\lim\limits_{m_0\to\infty}\left(\log_2\frac{m_0\Omega_{b_2}}{|\mathcal{G}_2|}+\left(1-\frac{\delta_{b_2}}{m_0}\cdot\frac{m_0\Omega_{b_2}}{|\mathcal{G}_2|}\right)\log_2(k+1)\right)\\
           =&0,
     \end{aligned}
          \end{displaymath}
which means that $\lim\limits_{m_0\to\infty}\Delta(\textbf{V}^\prime_\mathcal{B})$=0. By the same discussion as in Theorem \ref{Th: our perfectness1}, we get that our scheme 2 is asymptotically perfect by Lemma \ref{le: loss entropy of CHSS}.
\end{proof}

\begin{theorem}\label{Theorem: summary2}
Let $\mathcal{P}$ be a set of $n$ participants, and it is partitioned into $u$ disjoint subsets $\mathcal{P}_1, \mathcal{P}_2,\ldots, \mathcal{P}_u$. For a threshold sequence $t_1, t_2, \ldots, t_u$ satisfying $1\leq t_1< t_2<\cdots<t_u\leq n$ and $t_{\ell}\leq |\mathcal{P}_{\ell}|$ for $\ell\in [u]$, our scheme 2 is a secure and asymptotically perfect CHSS scheme. Moreover, if $L=\{m_{0},m_{1},\ldots,m_{n}\}$ is a $1$-compact sequence of co-primes, then our scheme 1 is an asymptotically ideal CHSS scheme.
\end{theorem}

\begin{proof}
  By Definition \ref{def: Multilevel secret sharing of conjunctive}, Theorem \ref{Theorem:Correctness of our scheme2} and Theorem \ref{Th: our perfectness2}, our scheme 2 is a secure and asymptotically perfect CHSS scheme. From the proof of Theorem \ref{Theorem: summary1}, it holds that the information $\rho_2$ of our scheme 2 satisfies $\lim\limits_{m_0\to\infty}\rho_2=1$. Hence, by Definition \ref{def: Asymptotically Ideal HSS}, our scheme 2 is an asymptotically ideal CHSS scheme.
\end{proof}

\section{Conclusion}\label{Sec: conclusion}
In this work, we construct asymptotically ideal DHSS and CHSS schemes based on the Chinese Remainder Theorem over integer rings. Compared to prior CRT-based constructions, our scheme achieves security, an information rate asymptotically approaching $1$, and computational efficiency.



\begin{thebibliography}{99}

\bibitem{Shamir1979} A. Shamir, How to share a secret, Communications of the ACM, vol.22, no.11, pp.612-613, 1979.

\bibitem{Blakley1979} G. R. Blakley, Safeguarding cryptographic keys, in: 1979 International Workshop on Managing Requirements Knowledge, MARK, New York, NY, USA, pp. 313-318. IEEE, 1979.

\bibitem{Simmons1988} G. J. Simmons, How to (really) share a secret, in: Advances in Cryptology-Crypto'88, Santa Barbara, California, USA. Lecture Notes in Computer Science, vol.403, pp.390-448. Springer, 1988.

\bibitem{Tassa2007} T. Tassa, Hierarchical threshold secret sharing, Journal of Cryptology, vol.20, no.2, pp.237-264, 2007.

\bibitem{Chenqi2022} Q. Chen, C. Tang, and Z. Lin, Efficient explicit constructions of multipartite secret sharing schemes, IEEE Transactions on Information Theory, vol.68, no.1, pp.601-631, 2022.

\bibitem{YuanJiaotong2022}J. Yuan, J. Yang, C. Wang, X. Jia, F. Fu, G. Xu, A new efficient hierarchical multi-secret sharing scheme based on linear homogeneous recurrence relations, Information Sciences, vol. 592, pp.36-49,2022.

\bibitem{Harn-Miao2014} L. Harn, and F. MIAO, Multilevel threshold secret sharing based on the Chinese remainder theorem, Information Processing Letters, vol.114, no.9, 2014. 504-509.

\bibitem{Oguzhan-Ersoy2016} O. Ersoy, K. Kaya, and K. Kaskaloglu, Multilevel Threshold Secret and Function Sharing based on the Chinese Remainder Theorem, arXiv: 1605.07988, https://arxiv.org/abs/1605.07988.

\bibitem{Tiplea2021}F. L. Tiplea, and C. C. Dr$\breve{\mathrm{a}}$gan, Asymptotically ideal Chinese remainder theorem-based secret sharing schemes for multilevel and compartmented
access structures, IET Information Security, vol. 15, no. 4, pp. 282-296, 2021.

\bibitem{Yangjing2024} J. Yang, S.-T. Xia, X. Wang, J. Yuan, and F.-W. Fu, A perfect ideal hierarchical secret sharing scheme based on the CRT for Polynomial Rings, in: 2024 IEEE International Symposium on Information Theory (ISIT), Athens, Greece, pp.321-326. IEEE, 2024.

\bibitem{Hongjuarx} H. Li, J. Ding, F. Miao, C.Wang, C. Shu, Novel CRT-based Asymptotically Ideal Disjunctive Hierarchical Secret Sharing Scheme, 2026, arXiv:2603.16267.

\bibitem{CRTdefinition} H. Cohen, A Course in Computational Algebraic Number Theory, 4thed., Grad. Texts Math, Springer-Verlag, 2000.

\bibitem{Tiplea2018} C. C. Dr$\breve{\mathrm{a}}$gan, and F. L. Tiplea, On the asymptotic idealness of the Asmuth-Bloom threshold secret sharing scheme, Information Sciences, Volumes 463¨C464, pp. 75-85, 2018.


\bibitem{Ning2018} Y. Ning, F. Miao, W. Huang, K. Meng, Y. Xiong, and X. Wang, Constructing ideal secret sharing schemes based on Chinese Remainder Theorem, in: Advances in Cryptology - {ASIACRYPT} 2018, Brisbane, QLD, Australia. Lecture Notes in Computer Science, vol. 11274, pp. 310-331. Springer, 2018.

\bibitem{Quisquater2002} M. Quisquater, B. Preneel, and J. Vandewalle, On the security of the threshold scheme based on the Chinese remainder theorem, in: Public Key Cryptography, PKC 2002. Lecture Notes in Computer Science, vol. 2274, pp.199-210. Springer, Berlin, Heidelberg, 2002.

\bibitem{Brickell1989} E. F. Brickell, Some ideal secret sharing schemes, in: Advances in Cryptology-Eurocrypt'89, Houthalen, Belgium. Lecture Notes in Computer Science, vol.434, pp.468-475. Springer, 1989.

\bibitem{Asmuth-Bloom1983}C. Asmuth and J. Bloom, A modular approach to key safeguarding, IEEE Transactions on Information Theory, vol. 30, no. 2, pp. 208-210, 1983.

\bibitem{Mo2023} S. Mo, Ideal hierarchical secret sharing and lattice path matroids, Designs, Codes and Cryptography, vol.91, no.4, pp. 1335-1349, 2023.

\end{thebibliography}
\end{document}